\documentclass[10pt]{article}
\usepackage[inner=2cm,outer=2cm]{geometry} 
\usepackage[english]{babel}
\usepackage[utf8]{inputenc}
\usepackage[T1]{fontenc}  
\usepackage{enumerate}
\usepackage{color}
\usepackage[T1]{fontenc}
\usepackage{subfigure}
\usepackage{dsfont}
\usepackage[final]{graphicx}
\usepackage[T1]{fontenc}
\usepackage{amsmath}
\usepackage{mathtools}
\usepackage{amsthm}
\usepackage{amstext}
\usepackage{amssymb}
\usepackage{mathrsfs}
\usepackage{cite}
\usepackage{mathtools}
\usepackage{tikz}
\usetikzlibrary{arrows,positioning}
\usepackage{xfrac}
\usepackage{cancel}

\newcommand{\me}{\mathrm{e}}
\newcommand{\nn}{\nonumber}
\newcommand{\mm}{m}

\newtheorem{thm}{Theorem}
\newtheorem{cor}{Corollary}
\newtheorem{fact}{Fact}
\newtheorem{prop}{Proposition}

\newtheorem{rem}{Remark}

\allowdisplaybreaks

\begin{document}
\sloppy
\title{On the Number of Bins in Equilibria for Signaling Games
}
\author{Serkan~Sar{\i}ta\c{s}$^1$ \and Philippe Furrer$^2$ \and Sinan~Gezici$^3$ \and Tam\'{a}s Linder$^4$ \and Serdar~Y\"uksel$^4$}
\date{%
	$^1$Division of Decision and Control Systems, KTH Royal Institute of Technology\\ SE-10044, Stockholm, Sweden. Email: saritas@kth.se.\\%
	$^2$Oliver Wyman Inc. \\  M5J 0A1, Toronto, Ontario, Canada. Email:phil.furrer@oliverwyman.com\\
	$^3$Department of Electrical and Electronics Engineering, Bilkent University\\ 06800, Ankara, Turkey. Email: gezici@ee.bilkent.edu.tr.\\
	$^4$Department of Mathematics and Statistics, Queen's University\\ K7L 3N6, Kingston, Ontario, Canada.  Email:  \{linder, yuksel\}@mast.queensu.ca.\\[2ex]%
}
\maketitle

\begin{abstract}
We investigate the equilibrium behavior for the decentralized quadratic cheap talk problem in which an encoder and a decoder, viewed as two decision makers, have misaligned objective functions. In prior work, we have shown that the number of bins under any equilibrium has to be at most countable, generalizing a classical result due to Crawford and Sobel who considered sources with density supported on $[0,1]$. In this paper, we refine this result in the context of exponential and Gaussian sources. For exponential sources, a relation between the upper bound on the number of bins and the misalignment in the objective functions is derived, the equilibrium costs are compared, and it is shown that there also exist equilibria with infinitely many bins under certain parametric assumptions. For Gaussian sources, it is shown that there exist equilibria with infinitely many bins.
\end{abstract}

\section{Introduction}

Signaling games and cheap talk are concerned with a class of Bayesian games where a privately informed player (encoder or sender) transmits information (signal) to another player (decoder or receiver), who knows the probability distribution of the possible realizations of the private information of the encoder. In these games/problems, the objective functions of the players are not aligned, unlike in classical communication problems. The cheap talk problem was introduced in the economics literature by Crawford and Sobel \cite{SignalingGames}, who obtained the striking result that under some technical conditions on the cost functions, the cheap talk problem only admits equilibria that involve quantized encoding policies. This is in significant contrast to the usual communication/information theoretic case where the objective functions are aligned. Therefore, as indicated in \cite{SignalingGames}, the similarity of players' interests (objective functions) indicates how much and to which extent the information should be revealed by the encoder; namely, the message about the private information should be strategically designed and transmitted by the encoder. In this paper, we discuss extensions and generalizations of strategic information transmission and cheap talk by focusing on exponential and Gaussian distributions (rather than sources with a density supported on $[0,1]$ as studied in \cite{SignalingGames}), and characterize the equilibrium solutions and properties for these kind of sources.

\subsection{Problem Definition}

The focus of this paper is to address the following problems:

\subsubsection{Number of Bins}
In our previous work \cite{tacWorkArxiv}, we showed that, since the distances between the optimal decoder actions are lower bounded by \cite[Theorem 3.2]{tacWorkArxiv}, the quantized nature of an equilibrium holds for arbitrary scalar sources, rather than only for sources with a density supported on $[0,1]$ as studied in the seminal paper by Crawford and Sobel \cite{SignalingGames}. Hence, for bounded sources, it can easily be deduced that the number of bins at the equilibrium must be bounded. For example, for a uniform source on $[0,1]$ and quadratic objective functions, \cite{SignalingGames} provides an upper bound on the number of quantization bins as a function of the bias $b$. Accordingly, for unbounded sources, the following problems are of interest: 
\begin{itemize}
	\item For unbounded sources, either one-sided or two-sided, is there an upper bound on the number of bins at the equilibrium as a function of bias $b$? As a special case, is it possible to have only a non-informative equilibrium; i.e., the upper bound on the number of bins is one?
	\item Is it possible to have an equilibrium with infinitely many bins?
\end{itemize}
At this point, one can ask why bounding the number of bins is important. Finding such bound is useful since if one can show that there only exists a finite number of bins, and if for every bin there is a finite number of distinct equilibria, then the total number of equilibria would be finite; this will allow for a feasible setting where the decision makers can coordinate their policies.

Furthermore, in a recent work, where we generalized signaling games and cheap talk problems to dynamic / multi-stage setups, a crucial property that allowed the generalization was the assumption that the number of bins for each stage equilibrium, conditioned on the past actions, is uniformly bounded \cite[Theorem 2.4]{dynamicGameArxiv}. In view of this, showing that the number of bins is finite would be a useful technical result. 

\subsubsection{Equilibrium Selection}
Attaining the upper bound $N$ on the number of bins at the equilibrium implies that there exists at least one equilibrium with $1,2,\ldots\,N$ bins due to \cite[Theorem 1]{SignalingGames}, and thus, a new question arises: among these multiple equilibria, which one is preferred by the players? Results in \cite{SignalingGames} show that an equilibrium with more bins is preferable for both the encoder and the decoder for any source with a density bounded on $[0,1]$. Accordingly, for more general sources, one can ask that 
\begin{itemize}
	\item if there exist more than one equilibrium, which one of these should be selected by the encoder and decoder?
	\item to what extent can one argue that more bins lead to better performance?
\end{itemize}
Indeed, it is important to know whether in general a higher number of bins implies more desirable equilibria. If such a monotonic behavior holds for a class of sources, then both players will prefer to have an equilibrium with the highest number of bins.

\subsection{Preliminaries}

In cheap talk, there are two players with misaligned objective functions. An informed player (encoder) knows the value of an $\mathbb{M}$-valued random variable $M$ and transmits an $\mathbb{X}$-valued random variable $X$ to another player (decoder), who generates his $\mathbb{M}$-valued decision $U$ upon receiving $X$. The policies of the encoder and decoder are assumed to be deterministic; i.e., $x=\gamma^e(m)$ and $u=\gamma^d(x)=\gamma^d(\gamma^e(m))$. Let $c^e(m,u)$ and $c^d(m,u)$ denote the cost functions of the encoder and the decoder, respectively, when the action $u$ is taken for the corresponding message $m$. Then, given the encoding and decoding policies, the encoder's induced expected cost is $J^e\left(\gamma^e,\gamma^d\right) = \mathbb{E}\left[c^e(M, U)\right]$, whereas the decoder's induced expected cost is $J^d\left(\gamma^e,\gamma^d\right) = \mathbb{E}\left[c^d(M, U)\right]$. Here, we assume real valued random variables and quadratic cost functions; i.e., $\mathbb{M}=\mathbb{X}=\mathbb{R}$, $c^e\left(m,u\right) = \left(m-u-b\right)^2$ and $c^d\left(m,u\right) = \left(m-u\right)^2$, where $b$ denotes a bias term which is common knowledge between the players. We assume the simultaneous-move game; i.e., the encoder and the decoder announce their policies at the same time. Then a pair of policies $(\gamma^{*,e}, \gamma^{*,d})$ is said to be a Nash equilibrium \cite{basols99} if
\begin{align}
	\begin{split}
		J^e(\gamma^{*,e}, \gamma^{*,d}) &\leq J^e(\gamma^{e}, \gamma^{*,d})  \text{ for all } \gamma^e \in \Gamma^e \,,\\
		J^d(\gamma^{*,e}, \gamma^{*,d}) &\leq J^d(\gamma^{*,e}, \gamma^{d})  \text{ for all }\gamma^d \in \Gamma^d \,,
		\label{eq:nashEquilibrium}
	\end{split}
\end{align}
where $\Gamma^e$ and $\Gamma^d$ are the sets of all deterministic (and Borel measurable) functions from $\mathbb{M}$ to $\mathbb{X}$ and from $\mathbb{X}$ to $\mathbb{M}$, respectively. As observed from the definition in \eqref{eq:nashEquilibrium}, under the Nash equilibrium, each individual player chooses an optimal strategy given the strategy chosen by the other player. 

Due to results obtained in \cite{SignalingGames} and \cite{tacWorkArxiv}, we know that the encoder policy consists of convex cells (bins) at a Nash equilibrium\footnote{We note that, unlike Crawford and Sobel's simultaneous Nash equilibrium formulation, if one considers a Stackelberg formulation (see \cite[p.133]{basols99} for a definition), then the problem would reduce to a classical communication problem since the encoder would be committed a priori and the equilibrium would not be quantized; i.e., there exist affine equilibria \cite{tacWorkArxiv,dynamicGameArxiv,bayesianPersuasion,CedricWork,akyolITapproachGame,omerHierarchial,frenchGroup}.}. Now consider an equilibrium with $N$ bins, and let the $k$-th bin be the interval $[\mm_{k-1},\mm_k)$ with $\mm_0<\mm_1<\ldots<\mm_N$ and let $l_k$ denote the length of the $k$-th bin; i.e., $l_k=\mm_k-\mm_{k-1}$ for $k=1,2,\ldots,N$ (Note that $\mm_0=0$ and $\mm_N=+\infty$ for an exponential source, whereas $\mm_0=-\infty$ and $\mm_N=+\infty$ for a Gaussian source). By \cite[Theorem 3.2]{tacWorkArxiv}, at the equilibrium, decoder's best response to encoder's action is characterized by 
\begin{align} \label{centroid}
	u_k = \mathbb{E}[M|\mm_{k-1}\leq M<\mm_k] 
\end{align}
for the $k$-th bin; i.e., the optimal decoder action is the centroid for the corresponding bin. From  encoder's point of view, the best response of the encoder to decoder's action is determined by the nearest neighbor condition \cite[Theorem 3.2]{tacWorkArxiv} as follows:
\begin{align} 
	u_{k+1}-\mm_k = (\mm_k - u_k) - 2b \Leftrightarrow \mm_k = \frac{u_k + u_{k+1}}{2}+b \,. 
	\label{eq:centroidBoundaryEq}
\end{align}
Due to the definition of Nash equilibirum \cite{basols99}, these best responses in \eqref{centroid} and \eqref{eq:centroidBoundaryEq} must match each other, and only then the equilibrium can be characterized; i.e., for a given number of bins, the positions of the bin edges are chosen by the encoder, and the centroids are determined by the decoder. Alternatively, the problem can be considered as a quantization game in which the boundaries are determined by the encoder and the reconstruction values are determined by the decoder.

Note that at the equilibrium of this quantization game, the relation between the encoder cost and the decoder cost can be expressed as follows:
\begin{align*}
	J^e(\gamma^{*,e}, \gamma^{*,d}) &= \sum_{i=1}^{N}\mathrm{Pr}(\mm_{i-1}<M<\mm_i)\mathbb{E}\left[(M-\mathbb{E}[M|\mm_{i-1}\leq M<\mm_i]-b)^2|\mm_{i-1}<M<\mm_i\right] \nn\\
	&=\sum_{i=1}^{N}\mathrm{Pr}(\mm_{i-1}<M<\mm_i)\left(\mathbb{E}\left[(M-\mathbb{E}[M|\mm_{i-1}\leq M<\mm_i])^2|\mm_{i-1}<M<\mm_i\right]+b^2\right) \nn\\
	&= J^d(\gamma^{*,e}, \gamma^{*,d})+b^2 \,.
\end{align*}
Since the difference between the encoder cost and the decoder cost is always $b^2$ regardless of the number of bins, the equilibrium preferences (i.e., which equilibrium to select) for the encoder and the decoder are aligned under the quadratic cost assumption.

Based on the above, the problems we consider in this paper can be formulated more formally as follows:

\subsubsection{Number of Bins}
For a given finite (or infinite) $N$, does there exist an equilibrium; i.e., is it possible to find the optimal encoder actions (the boundaries of the bins) $\mm_0,\mm_1,\ldots,\mm_N$ and decoder actions (the centroids of the bins) $u_1,u_2,\ldots,u_N$ which satisfy \eqref{centroid} and \eqref{eq:centroidBoundaryEq} simultaneously? Here two possible different methods are:
\begin{enumerate}
	\item[(i)] Lloyd's Method I: After the initial selection of $\mm_0,\mm_1,\ldots,\mm_N$, determine $u_1,u_2,\ldots,u_N$ by \eqref{centroid}, and after updating $u_1,u_2,\ldots,u_N$, find the corresponding $\mm_0,\mm_1,\ldots,\mm_N$ by \eqref{eq:centroidBoundaryEq}. Then, continue this iteration. For this approach, the convergence of this Lloyd-Max iteration is the key issue.
	\item[(ii)] Fixed-point approach: By combining \eqref{centroid} and \eqref{eq:centroidBoundaryEq}, 
	\[\mm_k = \frac{\mathbb{E}[M|\mm_{k-1}\leq M<\mm_k] + \mathbb{E}[M|\mm_{k}\leq M<\mm_k+1]}{2}+b\] 
	is obtained for $k=1,2,\ldots,N$. Then, the problem reduces to determining whether there exists a fixed vector $\mm_0,\mm_1,\ldots,\mm_N$ satisfying these equations. 
\end{enumerate}

\subsubsection{Equilibrium Selection} 
Let $J^{d,N}$ denote the decoder cost at the equilibrium with $N$ bins. Then, is it true that $J^{d,N}>J^{d,N+1}$ for any finite $N$, or even, is $J^{d,N}>J^{d,\infty}$ if an equilibrium with infinitely many bins exists? 

\subsection{Related Literature}

Cheap talk and signaling game problems find applications in networked control systems when a communication channel/network is present among competitive and non-cooperative decision makers \cite{basols99,misBehavingAgents}. Also, there have been a number of related results in the economics and control literature in addition to the seminal work by Crawford and Sobel, which are reviewed in \cite{tacWorkArxiv,dynamicGameArxiv} (see \cite{sobelSignal} for an extensive survey).

The quantized nature of the equilibrium makes game theory connected with the quantization theory. For a comprehensive survey regarding the history of  quantization and results on the optimality and convergence properties of different quantization techniques (including Lloyd's methods), we refer to \cite{quantizationSurvey}. In particular, \cite{Fleischer64} shows that, for sources with a log-concave density, Lloyd’s Method I converges to the unique optimal quantizer. It was shown in \cite{trushkin82} and \cite{kieffer83} that Lloyd’s Method I converges to the globally optimal quantizer if the source density is continuous and log-concave, and if the error weighting function is convex and symmetric. For sources with bounded support, the condition on the source was relaxed to include all continuous and positive densities in \cite{wu92}, and convergence of Lloyd’s Method I to a (possibly) locally optimal quantizer was proved. The number of bins of an optimal entropy-constrained quantizer is investigated in \cite{gyorgyLinder2003}, and conditions under which the number of bins is finite or infinite are presented. As an application to smart grids, \cite{csLloydMax} considers the design of signaling schemes between a consumer and an electricity aggregator with finitely many messages (signals); the best responses are characterized and the maximum number of messages (i.e., quantization bins) are found using Lloyd's Method II via simulation.

The existence of multiple quantized equilibria necessitates a theory to specify which equilibrium point is the solution of a given game. Two different approaches are taken to achieve a unique equilibrium. 
One of them reduces the multiplicity of equilibria by requiring that off-the-equilibrium-path beliefs satisfy an additional restriction (e.g., by shrinking the set of players' rational choices) \cite{eqSelectionRefinement}, \cite{sobelSignal}. As introduced in \cite{eqSelectionBook}, the other approach presents a theory that selects a unique equilibrium point for each finite game as its solution; i.e., one and only one equilibrium points out of the set of all equilibrium points of this kind (e.g., see \cite{eqSelectionSignaling} for the application). 

\subsection{Contributions}
\begin{enumerate}
	\item[(i)] Under the exponential source assumption with a negative bias; i.e., $b<0$, we obtain an upper bound on the number of bins at the equilibrium and show that the equilibrium cost reduces as the number of bins increases.
\item[(ii)] Under the exponential source assumption with a positive bias; i.e., $b>0$, we prove that there exists a unique equilibrium with $N$ bins for any $N\in\mathbb{N}$ and there is no upper bound on the number of bins; in fact, there exist equilibria with infinitely many bins. Further, the equilibrium cost achieves its minimum at the equilibrium with infinitely many bins. 
\item[(iii)] Under the Gaussian source assumption, we show that there always exist equilibria with infinitely many bins regardless of the value of $b$.
\end{enumerate}

\section{Exponential Distribution} \label{EXP}

In this section, the source is assumed to be exponential and the number of bins at the equilibria is investigated. Before delving into the technical results, we observe the following fact:
\begin{fact} \label{fact:exponential}
	Let $M$ be an exponentially distributed r.v. with a positive parameter $\lambda$: i.e., the probability distribution function (PDF) of $M$ is $\mathsf{f}(m)=\lambda\me^{-\lambda m}$ for $m\geq0$. The expectation and the variance of an exponential r.v. truncated to the interval $[a,b]$ are $\mathbb{E}[M|a<M<b]={1\over\lambda}+a-{b-a\over\me^{\lambda (b-a)}-1}$ and $\mathrm{Var}\left(M|a<M<b\right)={1\over\lambda^2} - {(b-a)^2\over\me^{\lambda (b-a)}+\me^{-\lambda (b-a)}-2}$, respectively.
\end{fact}
\begin{proof}
	Consider the following integral:
\begin{align}
	\int \lambda m\me^{-\lambda m}\mathrm{d}m \hspace*{1.0em} &\stackrel{\mathclap{\substack{{s=\lambda m} \\ {\mathrm{d}s=\lambda \mathrm{d}m}}}}{=} \hspace*{1.0em}  \int{1\over\lambda}s\me^{-s}\mathrm{d}s \overset{\substack{{u=s,\,\mathrm{d}v={\me^{-s}\mathrm{d}s/\lambda}} \\ {\mathrm{d}u=\mathrm{d}s,\,v={-\me^{-s}/\lambda}}}}{=} {-s\me^{-s}\over\lambda}-\int {-\me^{-s}\over\lambda}\mathrm{d}s \nn\\
	&= {-s\me^{-s}\over\lambda}-{\me^{-s}\over\lambda} 
	\overset{s=\lambda m}{=}  -m\me^{-\lambda m}-{\me^{-\lambda m}\over\lambda} \,.
	\label{eq:expMeanIntegral}
\end{align}
Then, the expectation of an exponential r.v. truncated to $[a,b]$ will be
\begin{align}
	\mathbb{E}[M|a<M<b] &= \int_{a}^{b} m {\lambda\me^{-\lambda m}\over\int_{a}^{b}\lambda\me^{-\lambda m}}\mathrm{d}m = {\int_{a}^{b} m\lambda\me^{-\lambda m}\mathrm{d}m\over\int_{a}^{b}\lambda\me^{-\lambda m}\mathrm{d}m} = {\left(-m\me^{-\lambda m}-{\me^{-\lambda m}\over\lambda}\right)\bigg\rvert_a^b \over -\me^{-\lambda m}\bigg\rvert_a^b} \nn\\
	&={-b\me^{-\lambda b}-{\me^{-\lambda b}\over\lambda} + a\me^{-\lambda a}+{\me^{-\lambda a}\over\lambda}\over-\me^{-\lambda b}+\me^{-\lambda a} } = {1\over\lambda}+{a\me^{\lambda b}-b\me^{\lambda a}\over \me^{\lambda b}-\me^{\lambda a}} \nn\\
	&={1\over\lambda}+a-{\me^{\lambda a}(b-a)\over\me^{\lambda b}-\me^{\lambda a}} = {1\over\lambda}+a-{b-a\over\me^{\lambda (b-a)}-1} \,.
	\label{eq:expCentroid}
\end{align}

Now consider the following integral:
\begin{align}
	\int \lambda m^2\me^{-\lambda m}\mathrm{d}m \hspace*{4.5em} & \stackrel{\mathclap{\substack{{u=\lambda m^2,\,\mathrm{d}v={\me^{-\lambda m}\mathrm{d}m}} \\ {\mathrm{d}u=2\lambda m\mathrm{d}m,\,v={-\me^{-\lambda m}/\lambda}}}}}{=} \hspace*{4.5em} \lambda m^2{-\me^{-\lambda m}\over\lambda}-\int {-\me^{-\lambda m}\over\lambda}2\lambda m\mathrm{d}m \nn\\
	&= -m^2\me^{-\lambda m} + {2\over\lambda}\int \lambda m\me^{-\lambda m}\mathrm{d}m  \nn\\
	& \stackrel{\mathclap{\text{(a)}}}{=}  -m^2\me^{-\lambda m} - {2m\me^{-\lambda m}\over\lambda}-{2\me^{-\lambda m}\over\lambda^2} \,.
\end{align}
Here, (a) holds due to \eqref{eq:expMeanIntegral}. Then, 
\begin{align}
	\mathrm{Var}\left(M|a<M<b\right) &= \mathbb{E}[M^2|a<M<b] - \left(\mathbb{E}[M|a<M<b]\right)^2 \nn\\
	&= \int_{a}^{b} m^2 {\lambda\me^{-\lambda m}\over\int_{a}^{b}\lambda\me^{-\lambda m}}\mathrm{d}m - \left({-b\me^{-\lambda b}-{\me^{-\lambda b}\over\lambda} + a\me^{-\lambda a}+{\me^{-\lambda a}\over\lambda}\over-\me^{-\lambda b}+\me^{-\lambda a} }\right)^2 \nn\\
	&= {-{\me^{-\lambda b}\over\lambda^2}\left(\lambda^2b^2+2\lambda b+2\right)+{\me^{-\lambda a}\over\lambda^2}\left(\lambda^2a^2+2\lambda a+2\right)\over-\me^{-\lambda b}+\me^{-\lambda a}} - {\left(-{\me^{-\lambda b}\over\lambda}\left(\lambda b+1\right)+{\me^{-\lambda a}\over\lambda}\left(\lambda a+1\right)\right)^2\over\left(-\me^{-\lambda b}+\me^{-\lambda a}\right)^2} \nn\\
	&= {{\me^{-2\lambda b}\over\lambda^2}\left(\lambda^2b^2+2\lambda b+2\right)+{\me^{-2\lambda a}\over\lambda^2}\left(\lambda^2a^2+2\lambda a+2\right)\over\left(-\me^{-\lambda b}+\me^{-\lambda a}\right)^2} - {{\me^{-\lambda (a+b)}\over\lambda^2}\left(\lambda^2a^2+\lambda^2b^2+2\lambda a+2\lambda b+4\right)\over\left(-\me^{-\lambda b}+\me^{-\lambda a}\right)^2} \nn\\
	&\qquad - {{\me^{-2\lambda b}\over\lambda^2}\left(\lambda^2b^2+2\lambda b+1\right)+{\me^{-2\lambda a}\over\lambda^2}\left(\lambda^2a^2+2\lambda a+1\right)\over\left(-\me^{-\lambda b}+\me^{-\lambda a}\right)^2} + {{\me^{-\lambda (a+b)}\over\lambda^2}\left(2\lambda^2ab+2\lambda a+2\lambda b+2\right)\over\left(-\me^{-\lambda b}+\me^{-\lambda a}\right)^2} \nn\\
	&= {{\me^{-2\lambda b}\over\lambda^2}+{\me^{-2\lambda a}\over\lambda^2} - {\me^{-\lambda (a+b)}\over\lambda^2}\left(\lambda^2 a^2 + \lambda^2 b^2 + 2 - 2\lambda^2ab\right) \over\left(-\me^{-\lambda b}+\me^{-\lambda a}\right)^2} \nn\\
	&= {{\me^{-2\lambda b}\over\lambda^2}+{\me^{-2\lambda a}\over\lambda^2} - {2\me^{-\lambda (a+b)}\over\lambda^2} - {\me^{-\lambda (a+b)}\over\lambda^2}\left(\lambda^2 (b-a)^2\right) \over\me^{-2\lambda b}+\me^{-2\lambda a}-2\me^{-\lambda (a+b)}} \nn\\
	&= {1\over\lambda^2} - {(b-a)^2\over\me^{-\lambda b + \lambda a}+\me^{-\lambda a + \lambda b}-2} \nn\\
	&= {1\over\lambda^2} - {(b-a)^2\over\me^{\lambda (b-a)}+\me^{-\lambda (b-a)}-2} \,.
	\label{eq:expTruncVariance}
\end{align}
\end{proof}
	
The following result shows the existence of an equilibrium with finitely many bins. Here, $\lfloor x\rfloor$ denotes the largest integer less than or equal to $x$.

\begin{prop} \label{prop1}
	Suppose $M$ is exponentially distributed with parameter $\lambda$. Then, for $b<0$, any Nash equilibrium is deterministic and can have at most $\lfloor -\frac{1}{2b\lambda} + 1\rfloor$ bins with monotonically increasing bin-lengths.
\end{prop}
\begin{proof}
	Since $u_N = \mathbb{E}[M|\mm_{N-1}\leq M\leq \mm_{N}=\infty] = \mm_{N-1} + {1\over\lambda}$, it follows that
\begin{align*}
\frac{1}{\lambda} = u_{N} - \mm_{N-1} 	&= (\mm_{N-1} - u_{N-1}) - 2b\\
&> (u_{N-1} - \mm_{N-2}) - 2b\\
&=  (\mm_{N-2} - u_{N-2}) - 2(2b)\\
&\vdots\\
&> u_1-\mm_0 - (N-1)(2b)\\
&> -(N-1)(2b) \,.
\end{align*}
	Here, the inequalities follow from the fact that the exponential PDF is monotonically decreasing. Hence, for $b<0$,
	\begin{align*}
		N < -{1\over2b\lambda} + 1 \Rightarrow N \leq \bigg\lfloor -{1\over2b\lambda} + 1 \bigg \rfloor \;.
	\end{align*}
	Now, consider the bin-lengths as follows:
	\begin{align}
		l_k &= \mm_k-\mm_{k-1} = (\mm_k-u_k) + (u_k-\mm_{k-1}) \geq (u_k-\mm_{k-1}) + (u_k-\mm_{k-1}) \nn\\
		&= (\mm_{k-1}-u_{k-1}-2b) + (\mm_{k-1}-u_{k-1}-2b) \nn\\
		&> (\mm_{k-1}-u_{k-1}-2b) + (u_{k-1}-\mm_{k-2}-2b) \nn\\
		&= (\mm_{k-1}-u_{k-1}) + (u_{k-1}-\mm_{k-2}) -4b = \mm_{k-1}-\mm_{k-2} - 4b = l_{k-1}-4b \nn\\
		&\Rightarrow l_k > l_{k-1} \,.\label{eq:expMonotonBinLengths}
	\end{align}
	Thus, the bin-lengths are monotonically increasing; i.e., $l_1<l_2<\ldots<l_{N-1}<l_N=\infty$.
\end{proof}

This is an important result as it provides us with a closed form expression for the maximum bit rate required by a certain system to operate at a steady state. For example, there can be at most one bin at the equilibrium (i.e., a non-informative equilibrium) if $ N \leq \big\lfloor -{1\over2b\lambda} + 1 \big \rfloor<2 \Leftrightarrow -{1\over2b\lambda}<1 \Leftrightarrow b<-{1\over2\lambda}$. However, this result does not characterize the equilibrium completely; i.e., it does not give a condition on the existence of an equilibrium with two or more bins. The following theorem characterizes the equilibrium with two bins, and forms a basis for equilibria with more bins:
\begin{thm}
	When the source has an exponential distribution with parameter $\lambda$, there exist only non-informative equilibria if and only if $b\leq-{1\over2\lambda}$. An equilibrium with at least two bins is achievable if and only if $b>-{1\over2\lambda}$. 
\end{thm}
\begin{proof}
	Consider the two bins $[0=\mm_0,\mm_1)$ and $[\mm_1,\mm_2=\infty)$. Then, the centroids of the bins (the decoder actions) are $u_1=\mathbb{E}[M|0<M<\overline{m}_1]={1\over\lambda}-{\mm_1\over\me^{\lambda\mm_1}-1}$ and $u_2=\mathbb{E}[M|\overline{m}_1<M<\infty]={1\over\lambda}+\mm_1$. In view of \eqref{eq:centroidBoundaryEq}, an equilibrium with these two bins exists if and only if
\begin{align}
	\mm_1 &= {u_1+u_2\over2}+b = {{1\over\lambda}-{\mm_1\over\me^{\lambda\mm_1}-1}+{1\over\lambda}+\mm_1\over2}+b \Rightarrow {\mm_1\over2}{\me^{\lambda\mm_1}\over\me^{\lambda\mm_1}-1}={1\over\lambda}+b \Rightarrow \me^{\lambda\mm_1}\left({1\over\lambda}+b-{\mm_1\over2}\right) = {1\over\lambda}+b \,.
	\label{eq:twoBinsEq}
\end{align}
Note that in \eqref{eq:twoBinsEq}, $\mm_1=0$ is always a solution; however, in order to have an equilibrium with two bins, we need a non-zero solution to \eqref{eq:twoBinsEq}; i.e., $\mm_1>0$. For this purpose, the Lambert $W$-function will be utilized. Although the Lambert $W$-function is defined for complex variables, we restrict our attention to the real-valued $W$-function; i.e., the $W$-function is defined as
\begin{align*}
	W(x\me^x) &= x \text{ for } x\geq 0 \,,\nn\\
	W_0(x\me^x) &= x \text{ for } -1\leq x<0 \,,\nn\\
	W_{-1}(x\me^x) &= x \text{ for } x\leq -1 \,.
\end{align*}
As it can be seen, for $x\geq0$, $W(x\me^x)$ is a well-defined single-valued function. However, for $x<0$, $W(x\me^x)$ is doubly valued, such as $W(x\me^x)\in(-{1\over\me},0)$ and there exist $x_1$ and $x_2$ that satisfy $x_1\me^{x_1}=x_2\me^{x_2}$ where $x_1\in(-1,0)$ and $x_2\in(-\infty,-1)$. In order to differentiate these values, the principal branch of the Lambert $W$-function is defined to represent the values greater than $-1$; e.g., $x_1=W_0(x_1\me^{x_1})=W_0(x_2\me^{x_2})$. Similarly, the lower branch of the Lambert $W$-function represents the values smaller than $-1$; e.g, $x_2=W_{-1}(x_1\me^{x_1})=W_{-1}(x_2\me^{x_2})$. Further, for $x=-1$, two branches of the $W$-function coincide; i.e., $-1=W_0(-{1\over\me})=W_{-1}(-{1\over\me})$. Regarding the definition above, by letting $t\triangleq2\lambda\left({\mm_1\over2}-{1\over\lambda}-b\right)$, the solution of \eqref{eq:twoBinsEq} can be found as follows:
\begin{align}
	\me^{t+2+2\lambda b}\left(-t\over2\lambda\right) = {1\over\lambda}+b \Rightarrow t\me^t = -(2+2\lambda b)\me^{-(2+2\lambda b)} \Rightarrow t = W_0\left(-(2+2\lambda b)\me^{-(2+2\lambda b)}\right) \,.
	\label{eq:twoBinsLambert}
\end{align}
Note that, in \eqref{eq:twoBinsLambert}, depending on the values of $-(2+2\lambda b)$, the following cases can be considered: 
\begin{enumerate}
	\item [(i)] \underline{$-(2+2\lambda b) \geq 0$} : $t\me^t = -(2+2\lambda b)\me^{-(2+2\lambda b)} \Rightarrow t = -(2+2\lambda b) \Rightarrow \mm_1=0$, which implies a non-informative equilibrium; i.e., an equilibrium with only one bin.
	\item [(ii)] \underline{$-1 < -(2+2\lambda b) < 0$} : Since $t\me^t = -(2+2\lambda b)\me^{-(2+2\lambda b)}$, there are two possible solutions: 
	\begin{enumerate}
		\item If $t=W_0\left(-(2+2\lambda b)\me^{-(2+2\lambda b)}\right)=-(2+2\lambda b)$, we have $\mm_1=0$, as in the previous case.
		\item If $t=W_{-1}\left(-(2+2\lambda b)\me^{-(2+2\lambda b)}\right) \Rightarrow t<-1 \Rightarrow -1>t=2\lambda\left({\mm_1\over2}-{1\over\lambda}-b\right)=\lambda\mm_1-2-2\lambda b > \lambda\mm_1-1 \Rightarrow \lambda\mm_1<0$, which is not possible.
	\end{enumerate} 
	\item [(iii)] \underline{$-(2+2\lambda b) = -1$} : Since $t\me^t = -(2+2\lambda b)\me^{-(2+2\lambda b)}$, there is only one solution, $t=-(2+2\lambda b)=-1\Rightarrow\mm_1=0$, which implies that the equilibrium is non-informative.
	\item [(iv)] \underline{$-(2+2\lambda b) < -1$} : Since $t\me^t = -(2+2\lambda b)\me^{-(2+2\lambda b)}$, there are two possible solutions:
	\begin{enumerate}
		\item If $t=W_{-1}\left(-(2+2\lambda b)\me^{-(2+2\lambda b)}\right)=-(2+2\lambda b)$, we have $\mm_1=0$; i.e., an equilibrium with only one bin.
		\item If $t=W_0\left(-(2+2\lambda b)\me^{-(2+2\lambda b)}\right)$, we have $-1<t<0 \Rightarrow -1<\lambda\mm_1-2-2\lambda b<0 \Rightarrow {1\over\lambda}+2b<\mm_1<{2\over\lambda}+2b$. Thus, if we have ${1\over\lambda}+2b>0 \Rightarrow b>-{1\over2\lambda}$, then $\mm_1$ must be positive, which implies the existence of an equilibrium with two bins.
	\end{enumerate} 
\end{enumerate}
Thus, as long as $b\leq-{1\over2\lambda}$, there exists only one bin at the equilibrium; i.e., there exist only non-informative equilibria; and the equilibrium with two bins can be achieved only if $b>-{1\over2\lambda}$. In this case, $\mm_1={1\over\lambda}W_0\left(-(2+2\lambda b)\me^{-(2+2\lambda b)}\right)+2\left({1\over\lambda}+b\right)$. Note that, since $-1<W_0(\cdot)<0$, the boundary between two bins lies within the interval ${1\over\lambda}+2b<\mm_1<{2\over\lambda}+2b$.
\end{proof}
	
Contrarily to the negative bias case, the number of bins at the equilibrium is not bounded when the bias is positive. The following theorem investigates the case in which $b>0$:
\begin{thm}
	When the source has an exponential distribution with parameter $\lambda$, for $b>0$ and any number of bins $N$, 
\begin{enumerate}
	\item [(i)] There exists a unique equilibrium,
	\item [(ii)] The bin-lengths are monotonically increasing.
\end{enumerate}
Further, since the two statements above hold for any $N\in\mathbb{N}$, there exists no upper bound on the number of bins at an equilibrium.
\label{thm:expPosBias}
\end{thm}
\begin{proof}
The proof consists of three main parts. After characterizing the equilibrium, the monotonicity of bin-lengths and the upper bound on the number of bins are investigated.

\underline{Part-I: Equilibrium Solution} : For the last two bins,
\begin{align}
	\mm_{N-1}&={u_{N-1}+u_N\over2}+b= {\left(\mm_{N-2}+{1\over\lambda}-{l_{N-1}\over\me^{\lambda l_{N-1}}-1}\right)+\left(\mm_{N-1}+{1\over\lambda}\right)\over2}+b\nn\\
	&\Rightarrow  {l_{N-1}}{\me^{\lambda l_{N-1}}\over\me^{\lambda l_{N-1}}-1}={2\over\lambda}+2b \label{eq:bin_N-1}\\
	&\Rightarrow l_{N-1}={1\over\lambda}W_0\left(-(2+2\lambda b)\me^{-(2+2\lambda b)}\right)+2\left({1\over\lambda}+b\right) \label{eq:bin_N-1_lambert}
\end{align}	
can be obtained. For the other bins; i.e., the $k$-th bin for $k=1,2,\ldots,N-2$, observe the following:    
\begin{align}
	u_{k+1} &- \mm_k = \mm_k - u_k - 2b = (\mm_k-\mm_{k-1}) - (u_k-\mm_{k-1}) - 2b \nn\\
	\Rightarrow & {1\over\lambda}-{l_{k+1}\over\me^{\lambda l_{k+1}}-1} = l_k - {1\over\lambda}+{l_{k}\over\me^{\lambda l_{k}}-1}-2b \nn\\
	\Rightarrow & l_{k}{\me^{\lambda l_{k}}\over\me^{\lambda l_k}-1} = {2\over\lambda}+2b - {l_{k+1}\over\me^{\lambda l_{k+1}}-1} \,.
	\label{eq:bin_k}
\end{align}
If we let $c_{k}\triangleq {2\over\lambda}+2b - {l_{k+1}\over\me^{\lambda l_{k+1}}-1}$, the solution to \eqref{eq:bin_k} is
\begin{align}
	l_{k} = {1\over\lambda}W_0\left(-\lambda c_{k}\me^{-\lambda c_{k}}\right)+c_{k} \,. \label{eq:bin_k_lambert}
\end{align}
It can be observed from \eqref{eq:bin_N-1_lambert} and \eqref{eq:bin_k_lambert} that the bin-lengths $l_1, l_2, \ldots, l_{N-1}$ have a unique solution, which implies that the bin edges have unique values as $\mm_0=0$, $\mm_k=\sum_{i=1}^{k} l_k$ for $k=1,2,\ldots,N-1$, and $\mm_N=\infty$.

In order to represent the solutions in a recursive form, define $g(l_k)\triangleq l_{k}{\me^{\lambda l_{k}}\over\me^{\lambda l_k}-1}$ and $h(l_k) \triangleq {l_{k}\over\me^{\lambda l_{k}}-1}$. Then, the recursions in \eqref{eq:bin_N-1} and \eqref{eq:bin_k} can be written as:
\begin{subequations}
	\begin{align}
		g(l_{N-1})	&={2\over\lambda}+2b \,,\label{eq:compactRecursionLast}\\
		g(l_k) 	&= {2\over\lambda}+2b-h(l_{k+1})\,, \qquad \mbox{for $k=1,2,\ldots,N-2$} \,.
		\label{eq:compactRecursion}
	\end{align}
\end{subequations}

\underline{Part-II: Monotonically Increasing Bin-Lengths} : The proof is based on induction. Before the induction step, in order to utilize \eqref{eq:compactRecursion}, we examine the structure of $g$ and $h$. First note that both functions are continuous and differentiable on $[0,\infty)$. Now, $g$ has the following properties:
\begin{itemize}	
	\item $g(0) = \lim_{s\to0} \frac{s\me^{\lambda s}}{\me^{\lambda s}-1} \stackrel{H}{=} \lim_{s\to0} \frac{1+\lambda s}{\lambda} = \frac{1}{\lambda} > 0\quad$  ($\stackrel{H}{=}$ represents l'H\^ospital's rule), 
	\item $\lim_{s\to\infty} g(s) = \lim_{s\to\infty} \frac{s\me^{\lambda s}}{\me^{\lambda s}-1} \stackrel{H}{=} \lim_{s\to\infty} \frac{1+\lambda s}{\lambda} = \infty$,
	\item $\frac{\mathrm{d}}{\mathrm{d}s}(g(s))|_{s=0} = \lim_{s\to0} \frac{\me^{\lambda s}(\me^{\lambda s}-\lambda s-1)}{(\me^{\lambda s} - 1)^2} \stackrel{H}{=} \lim_{s\to0} \frac{\lambda \me^{\lambda s}(2\me^{\lambda s}-\lambda s-2)}{2\lambda \me^{\lambda s}(\me^{\lambda s} - 1)} \stackrel{H}{=} \lim_{s\to0}\frac{2\lambda \me^{\lambda s} - \lambda}{2\lambda \me^{\lambda s}} = \frac{1}{2} > 0$,
	\item $\frac{\mathrm{d}}{\mathrm{d}s}(g(s)) = \frac{\me^{\lambda s}(\me^{\lambda s}-\lambda s-1)}{(\me^{\lambda s} - 1)^2} =  \frac{\me^{\lambda s}\sum_{k=2}^{\infty}{(\lambda s)^k\over k!}}{(\me^{\lambda s} - 1)^2} > 0$, \quad for any $s>0$.
\end{itemize}
All of the above imply that $g$ is a positive, strictly increasing and unbounded function on $\mathbb{R}_{\geq0}$.

Similarly, the properties of $h$ can be listed as follows:
\begin{itemize}
	\item $h(0) = \lim_{s\to0}\frac{s}{\me^{\lambda s}-1} \stackrel{H}{=} \lim_{s\to0} \frac{1}{\lambda \me^{\lambda s}} = \frac{1}{\lambda} >0$,
	\item$\lim_{s\to\infty} h(s) = \lim_{s\to\infty} \frac{s}{\me^{\lambda s}-1} \stackrel{H}{=} \lim_{s\to\infty} \frac{1}{\lambda \me^{\lambda s}} = 0$,
	\item $\frac{\mathrm{d}}{\mathrm{d}s}(h(s))|_{s=0} = \lim_{s\to0} -\frac{\me^{\lambda s}(\lambda s -1)+1}{(\me^{\lambda s}-1)^2} \stackrel{H}{=} \lim_{s\to0} {-\lambda^2 s \me^s \over 2(\me^{\lambda s}-1)\lambda \me^{\lambda s}} \stackrel{H}{=} \lim_{s\to0} -\frac{\lambda}{2\lambda\me^{\lambda s}} = -\frac{1}{2} <0$,
	\item $\frac{\mathrm{d}}{\mathrm{d}s}(h(s)) = -\frac{\me^{\lambda s}(\lambda s -1)+1}{(\me^{\lambda s}-1)^2}\stackrel{(a)}{<}0$, \quad for any $s>0,\\$
	where (a) follows from the fact that $\frac{\mathrm{d}}{\mathrm{d}s}(-\me^{\lambda s}(\lambda s -1)-1) = -\lambda^2s\me^{\lambda s}<0$ for any $s>0$, and $-\me^0(\lambda (0) - 1) - 1 = 0$.
\end{itemize}

All of the above imply that $h$ is a positive and strictly decreasing function on $\mathbb{R}_{\geq0}$.

Further, notice the following properties:
\begin{align}
	\begin{split}
		g(l_k)&=h(l_k)+l_k \,,\\
		g(l_k)&=l_{k}{\me^{\lambda l_{k}}\over\me^{\lambda l_k}-1}>l_k \,,\\
		h(l_k) &= {l_{k}\over\me^{\lambda l_{k}}-1} = {l_{k}\over\sum_{k=0}^{\infty}{(\lambda l_{k})^k\over k!}-1} = {l_{k}\over\lambda l_{k}+\sum_{k=2}^{\infty}{(\lambda l_{k})^k\over k!}}<{l_k\over\lambda l_k} = {1\over\lambda} \,.
		\label{eq:gh_Properties}
	\end{split}
\end{align}

Now consider the length of the $(N-2)$-nd bin. By utilizing the properties in \eqref{eq:gh_Properties} on the recursion in \eqref{eq:compactRecursion},
\begin{align}
	g(l_{N-2}) &= {2\over\lambda}+2b-h(l_{N-1}) = g(l_{N-1})-h(l_{N-1}) = l_{N-1} \nn\\ &\Rightarrow l_{N-1} = g(l_{N-2}) =l_{N-2}+h(l_{N-2}) \nn\\
	&\Rightarrow  l_{N-2} < l_{N-1} < l_{N-2}+{1\over\lambda}
\end{align}
is obtained. Similarly, for the $(N-3)$-rd bin, the following relations hold:
\begin{align}
	g(l_{N-3}) &= {2\over\lambda}+2b-h(l_{N-2}) = g(l_{N-2})+h(l_{N-1})-h(l_{N-2}) = l_{N-2}+h(l_{N-1}) \nn\\
	&g(l_{N-3})=l_{N-2}+h(l_{N-1})<l_{N-2}+h(l_{N-2})=g(l_{N-2}) \Rightarrow l_{N-3}<l_{N-2}\nn\\
	&l_{N-2}<l_{N-2}+h(l_{N-1})=g(l_{N-3})=l_{N-3}+h(l_{N-3})<l_{N-3}+{1\over\lambda}  \nn\\
	&\Rightarrow  l_{N-3} < l_{N-2} < l_{N-3}+{1\over\lambda} \,.
\end{align}
Now suppose that $l_{N-1}>l_{N-2}>\ldots>l_{k}$ is obtained. Then, consider the $(k-1)$-st bin:
\begin{align}
	g(l_{k-1}) &= {2\over\lambda}+2b-h(l_{k}) = g(l_{k})+h(l_{k+1})-h(l_{k}) = l_{k}+h(l_{k+1}) \nn\\
	&g(l_{k-1})=l_{k}+h(l_{k+1})<l_{k}+h(l_{k})=g(l_{k}) \Rightarrow l_{k-1}<l_{k}\nn\\
	&l_{k}<l_{k}+h(l_{k+1})=g(l_{k-1})=l_{k-1}+h(l_{k-1})<l_{k-1}+{1\over\lambda}  \nn\\
	&\Rightarrow  l_{k-1} < l_{k} < l_{k-1}+{1\over\lambda} \,.
\end{align}
Thus, the bin-lengths form a monotonically increasing sequence. 

\underline{Part-III: Number of Bins} : Consider the length of the $(N-1)$-st bin: Notice that, in \eqref{eq:bin_N-1_lambert}, since $b>0$, $-(2+2\lambda b)<-2<-1$ holds, and by the Lambert $W$-function, $t=W_0\left(-(2+2\lambda b)\me^{-(2+2\lambda b)}\right)$ such that $t\me^t=-(2+2\lambda b)\me^{-(2+2\lambda b)}$ and $-1<t<0$, which result in ${2\over\lambda}+2b>l_{N-1}>{1\over\lambda}+2b>{1\over\lambda}$; i.e., the $(N-1)$-st bin has a positive length. 

For the other bins, since $c_{k}={2\over\lambda}+2b-{l_{k+1}\over\me^{\lambda l_{k+1}}-1}={2\over\lambda}+2b-h(l_{k+1})>{2\over\lambda}+2b-{1\over\lambda}={1\over\lambda}+2b$ for $b>0$, we have $-\lambda c_k=-1-2\lambda b<-1$, which implies that $W_0\left(-\lambda c_k\me^{-\lambda c_k}\right)$ has a solution $t$ such that $t\me^t=-\lambda c_k\me^{-\lambda c_k}$ and $-1<t<0$. Hence, from \eqref{eq:bin_k_lambert}, $l_k={1\over\lambda}W_0\left(-\lambda c_{k}\me^{-\lambda c_{k}}\right)+c_{k}>{1\over\lambda}(-1)+{1\over\lambda}+2b=2b>0$ is obtained. This means that, for any given number of bins $N$, when $b>0$, an equilibrium with positive bin-lengths $l_1,l_2,\ldots,l_{N-2}$ is obtained.

To summarize the results, for every $N\in\mathbb{N}$ with $l_N=\infty$, there exists a solution $l_1, l_2,\ldots,l_N$ so that
\begin{enumerate}
	\item[(i)] these construct a unique equilibrium,
	\item[(ii)] each of these are non-zero,
	\item[(iii)] these form a monotonically increasing sequence.
\end{enumerate} 
\end{proof}

By following an approach similar to that in Theorem~\ref{thm:expPosBias}, the bounds in Proposition~\ref{prop1} can be refined as follows:
\begin{cor}
	\begin{itemize}
		\item [(i)] There exists an equilibrium with at least two bins if and only if $b>-{1\over2\lambda}$.
		\item [(ii)] There exists an equilibrium with at least three bins if and only if $b>-{1\over2\lambda}{\me-2\over\me-1}$.
	\end{itemize}
\end{cor}
\begin{proof}
\begin{enumerate}[(i)]
		\item In order to have an equilibrium with at least 2 bins, $l_{N-1}>0$ must be satisfied. From \eqref{eq:bin_N-1_lambert}, if $-(2+2\lambda b)<-1$ is satisfied, then the solution $l_{N-1}$ will be positive. Thus, if $b>-{1\over2\lambda}$, an equilibrium with at least 2 bins is obtained; otherwise; i.e., $b\leq-{1\over2\lambda}$, there exists only one bin at the equilibrium.
		\item In order to have an equilibrium with at least 3 bins, $l_{N-2}>0$ must be satisfied. From \eqref{eq:bin_k_lambert}, if $-\lambda c_{N-2}<-1$ is satisfied, then the solution to $l_{N-2}$ will be positive. Then,
		\begin{align}
			-\lambda c_{N-2} =& -\lambda\left({2\over\lambda}+2b-h(l_{N-1})\right)=-\lambda\left(g(l_{N-1})-h(l_{N-1})\right)=-\lambda l_{N-1} <-1 \nn\\
			\Rightarrow & l_{N-1}={1\over\lambda}W_0\left(-(2+2\lambda b)\me^{-(2+2\lambda b)}\right)+2\left({1\over\lambda}+b\right)>{1\over\lambda} \nn\\
			\Rightarrow & W_0\left(-(2+2\lambda b)\me^{-(2+2\lambda b)}\right) > -1-2\lambda b \,.
		\label{eq:threeBinsanalysis}
		\end{align}
		Let $t\triangleq W_0\left(-(2+2\lambda b)\me^{-(2+2\lambda b)}\right)$, then $t\me^t=-(2+2\lambda b)\me^{-(2+2\lambda b)}$ and $-1<t<0$. Then, from \eqref{eq:threeBinsanalysis}, since $t\me^t$ is increasing function of $t$ for $t>-1$,
		\begin{align}
			t>&-1-2\lambda b \Rightarrow t\me^t=-(2+2\lambda b)\me^{-(2+2\lambda b)} > -(1+2\lambda b)\me^{-(1+2\lambda b)} \nn\\
			\Rightarrow & 2+2\lambda b < (1+2\lambda b)\me \Rightarrow b > -{1\over2\lambda}{\me-2\over\me-1} \,.
		\end{align}
		Thus, if $b>-{1\over2\lambda}{\me-2\over\me-1}$, an equilibrium with at least three bins is obtained; otherwise; i.e., $b\leq-{1\over2\lambda}{\me-2\over\me-1}$, there can exist at most two bins at the equilibrium.
	\end{enumerate}
\end{proof}

Theorem~\ref{thm:expPosBias} shows that, when $b>0$, there exists an equilibrium with $N$ bins for any finite $N\in\mathbb{N}$. The following shows the existence of equilibria with infinitely many bins:
\begin{thm} \label{thm:expInfinite}
	For an exponential source and positive bias $b>0$, there exist equilibria with infinitely many bins. In particular, all bins must have a length of $l^*$, where $l^*$ is the solution to $g(l^*) = {2\over\lambda}+2b-h(l^*)$.
\end{thm}	
\begin{proof}
For any equilibrium, consider a bin with a finite length, let's say the $k$-th bin, and by utilizing \eqref{eq:compactRecursion} and \eqref{eq:gh_Properties}, we have the following inequalities:
\begin{align*}
	{2\over\lambda}+2b =& g(l_k)+h(l_{k+1}) = g(l_k) + g(l_{k+1}) - l_{k+1} > l_k +l_{k+1} - l_{k+1} = l_k \Rightarrow l_k < {2\over\lambda}+2b \,,\nn\\
	{2\over\lambda}+2b =& g(l_k)+h(l_{k+1}) = h(l_k)+l_k +h(l_{k+1}) < {1\over\lambda} + l_k + {1\over\lambda} = {2\over\lambda} + l_k 
	\Rightarrow l_k > 2b \,.
\end{align*}
Thus, all bin-lengths are bounded from above and below: $2b < l_k < {2\over\lambda}+2b$. Now consider the fixed-point solution of the recursion in \eqref{eq:compactRecursion}; i.e., $g(l^*)={2\over\lambda}+2b-h(l^*)$. Then, by letting $c\triangleq{2\over\lambda}+2b$,
\begin{align}
	l^*{\me^{\lambda l^*}\over\me^{\lambda l^*}-1} = c - {l^*\over\me^{\lambda l^*}-1} \Rightarrow l^*{\me^{\lambda l^*}+1\over\me^{\lambda l^*}-1} = c \Rightarrow (c-l^*)\me^{\lambda l^*}-(c+l^*)=0 \,.
	\label{eq:expInfFixedEq}
\end{align}
In order to investigate if \eqref{eq:expInfFixedEq} has a unique solution $l^*$ such that $2b<l^*<{2\over\lambda}+2b$, let $\Psi(s)\triangleq (c-s)\me^{\lambda s}-(c+s)$ for $s\in\left(2b,{2\over\lambda}+2b\right)$ and notice the following properties:
\begin{itemize}
	\item $\Psi(2b) = {2\over\lambda}\me^{2\lambda b}-\left({2\over\lambda}+4b\right) = {2\over\lambda}\left(\me^{2\lambda b}-1-2\lambda b\right)={2\over\lambda}\left(1+2\lambda b + \sum\limits_{k=2}^{\infty}{(2\lambda b)^k\over k!}-1-2\lambda b\right) = {2\over\lambda}\left( \sum\limits_{k=2}^{\infty}{(2\lambda b)^k\over k!}\right)>0$,
	\item $\Psi({2\over\lambda}+2b) = 0\times\me^{2+2\lambda b}-\left({4\over\lambda}+4b\right) = -\left({4\over\lambda}+4b\right)<0$,
	\item $\Psi^\prime(s)=\frac{\mathrm{d}}{\mathrm{d}s}(\Psi(s)) = \me^{\lambda s}\left(\lambda(c-s)-1\right)-1$,
	\item $\Psi^\prime(2b) = \me^{2\lambda b}\left(\lambda({2\over\lambda})-1\right)-1=\me^{2\lambda b}-1>0$,
	\item $\Psi^\prime({2\over\lambda}+2b) = \me^{2+2\lambda b}\left(\lambda\times0-1\right)-1=-\me^{2+2\lambda b}-1<0$,
	\item $\Psi^{\prime\prime}(s)=\frac{\mathrm{d}}{\mathrm{d}s}(\Psi^\prime) = \lambda\me^{\lambda s}\left(\lambda(c-s)-2\right)$, since $s\in\left(2b,{2\over\lambda}+2b\right)$ and $c={2\over\lambda}+2b$, we have $0<c-s<{2\over\lambda}\Rightarrow -2<\lambda(c-s)-2<0 \Rightarrow \Psi^{\prime\prime}(s)>0$.
\end{itemize}
All of the above implies that $\Psi(s)$ is a concave function of $s$ for $s\in\left(2b,{2\over\lambda}+2b\right)$, $\Psi(2b)>0$, $\Psi(s)$ reaches its maximum value on the interval $\left(2b,{2\over\lambda}+2b\right)$; i.e., when $\Psi^\prime(s^*)=0$, and $\Psi({2\over\lambda}+2b)<0$; thus, $\Psi(s)$ crosses the $s$-axis only once, which implies that  $\Psi(s)=0$ has a unique solution on the interval $\left(2b,{2\over\lambda}+2b\right)$. In other words, the fixed-point solution of the recursion in \eqref{eq:compactRecursion} is unique; i.e., $\Psi(l^*)=0$.

Hence, if the length of the first bin is $l^*$; i.e., $l_1=l^*$, then, all bins must have a length of $l^*$; i.e., $l_1=l_2=l_3=\ldots = l^*$. Thus, there exist equilibria with infinitely many equi-length bins.

Now, suppose that $l_1<l^*$. Then, by \eqref{eq:compactRecursion}, $h(l_2) = {2\over\lambda}+2b - g(l_1)$. Since $g$ in an increasing function, $l_1<l^*\Rightarrow g(l_1)<g(l^*)$. Let $g(l^*)-g(l_1)\triangleq\Delta>0$, then,
\begin{align}
	g(l^*) + h(l^*) = g(l_1) + h(l_2) = {2\over\lambda}+2b \Rightarrow \Delta = g(l^*)-g(l_1) = h(l_2) - h(l^*) \,.
	\label{eq:l_1SmallCase}
\end{align}
From Theorem~\ref{thm:expPosBias}, we know that $h(s)={s\over\me^{\lambda s}-1}$ is a decreasing function with $h^\prime(s)=-\frac{\me^{\lambda s}(\lambda s -1)+1}{(\me^{\lambda s}-1)^2}<0$ for $s>0$ and  $h^\prime(0)=-{1\over2}$. Slightly changing the notation, let $\widetilde{h}^\prime(s)=h^\prime({s\over\lambda})$; i.e., $\widetilde{h}^\prime(s)=\frac{\me^s-1-s\me^s}{(\me^s-1)^2}$. Then, $\widetilde{h}^{\prime\prime}(s)=\frac{\mathrm{d}}{\mathrm{d}s}(\widetilde{h}^\prime(s)) = -\frac{\me^s(\me^s-1)(2\me^s-s\me^s-s-2)}{(\me^s-1)^4}$. Now, let $\varrho(s)\triangleq2\me^s-s\me^s-s-2$, and observe the following:
\begin{align}
	\varrho(s) &= 2\me^s-s\me^s-s-2 \Rightarrow \varrho(0)=0\,,\nn\\
	\varrho^\prime(s)&=\frac{\mathrm{d}}{\mathrm{d}s}(\varrho(s)) = \me^s-s\me^s-1 \Rightarrow \varrho^\prime(0)=0\,, \nn\\
	\varrho^{\prime\prime}(s)&=\frac{\mathrm{d}}{\mathrm{d}s}(\varrho^\prime(s)) = -s\me^s \leq 0 \Rightarrow \varrho^{\prime\prime}(0)=0 \nn\\
	& \Rightarrow \varrho^\prime(s)<0 \text{ for } s>0 \Rightarrow \varrho(s)<0 \text{ for } s>0 \Rightarrow \widetilde{h}^{\prime\prime}(s)>0 \text{ for } s>0 \,.
	\label{eq:rhoNegative}
\end{align}
Thus, $\widetilde{h}^\prime(s)$ is an increasing function, which implies that $h^\prime(s)$ is also an increasing function. Since $h^\prime(0)=-{1\over2}$, $h^\prime(s)>-{1\over2}$ for $s>0$, it follows that ${h(l^*)-h(l_2)\over l^*-l_2}>-{1\over2} \Rightarrow {-\Delta\over l^*-l_2}>-{1\over2} \Rightarrow l^*-l_2>2\Delta$. From \eqref{eq:l_1SmallCase},
\begin{align}
	\Delta + \Delta &= (g(l^*)-g(l_1)) + (h(l_2) - h(l^*)) = g(l^*)-g(l_1) + (g(l_2)-l_2) - (g(l^*)-l^*) \nn\\
	&= g(l_2) - g(l_1) + \underbrace{l^* - l_2}_{>2\Delta} \Rightarrow g(l_2) - g(l_1) <0 \Rightarrow l_2<l_1 \,.
\end{align}
Proceeding similarly, $l^*>l_1>l_2>\ldots$ can be obtained. Now, notice that, since $h(l_k)$ is a monotone function and $2b < l_k < {2\over\lambda}+2b$, the recursion in \eqref{eq:compactRecursion} can be satisfied if 
\begin{align}
	g(l_k) = {2\over\lambda}+2b-h(l_{k+1}) \Rightarrow {2\over\lambda}+2b-h(2b) < g(l_k) < {2\over\lambda}+2b-h\left({2\over\lambda}+2b\right) \,.
\end{align}
Let $\underline{l}$ and $\overline{l}$ and defined as $g(\underline{l})={2\over\lambda}+2b-h(2b)$ and $g(\overline{l})={2\over\lambda}+2b-h\left({2\over\lambda}+2b\right)$, respectively. Thus, if $l_k \notin (\underline{l},\overline{l})$, then there is no solution to $l_{k+1}$ for the recursion in \eqref{eq:compactRecursion}. Since the sequence of bin-lengths is monotonically decreasing, there is a natural number $K$ such that $l_K>\underline{l}$ and $l_{K+1}\leq\underline{l}$, which implies that there is no solution to $l_{K+2}$. Thus, there cannot be any equilibrium with infinitely many bins if $l_1<l^*$.

A similar approach can be taken for $l_1>l^*$: Since $g$ is an increasing function, $l_1>l^*\Rightarrow g(l_1)>g(l^*)$. Let $g(l_1)-g(l^*)\triangleq\widetilde{\Delta}>0 \Rightarrow g(l_1)-g(l^*) = h(l^*) - h(l_2) = \widetilde{\Delta}$. Then, since $h^\prime(s)>-{1\over2}$ for $s>0$, ${h(l_2)-h(l^*)\over l_2-l^*}>-{1\over2} \Rightarrow {-\Delta\over l_2-l^*}>-{1\over2} \Rightarrow l_2-l^*>2\Delta$. From \eqref{eq:l_1SmallCase},
\begin{align}
	\widetilde{\Delta} + \widetilde{\Delta} &= (g(l_1)-g(l^*)) + (h(l^*) - h(l_2)) = g(l_1)-g(l^*) + (g(l^*)-l^*) - (g(l_2)-l_2) \nn\\
	&= g(l_1) - g(l_2) + \underbrace{l_2 - l^*}_{>2\Delta} \Rightarrow g(l_1) - g(l_2) <0 \Rightarrow l_1<l_2 \,.
\end{align}
Proceeding similarly, $l^*<l_1<l_2<\ldots$ can be obtained. Since the sequence of bin-lengths is monotonically increasing, there is a natural number $\widetilde{K}$ such that $l_{\widetilde{K}}<\overline{l}$ and $l_{\widetilde{K}+1}\geq\overline{l}$, which implies that there is no solution to $l_{\widetilde{K}+2}$. Thus, there cannot be any equilibrium with infinite number of bins if $l_1>l^*$. Notice that, it is possible to have an equilibrium with finite number of bins since for the last bin with a finite length, \eqref{eq:compactRecursionLast} is utilized. Further, it is shown that, at the equilibrium, any finite bin-length must be greater than or equal to $l^*$; i.e., $2b<l^*\leq l_k<{2\over\lambda}+2b$ must be satisfied.
\end{proof}

So far, we have shown that at the equilibrium there is an upper bound on the number of bins when $b<0$; i.e., there can exist only finitely many equilibria with finitely many bins. On the other hand, when $b>0$, there is no upper bound on the number of bins at the equilibrium, and even there exist equilibria with infinitely many bins. Therefore, at this point, it is interesting to examine which equilibrium is preferred by the decision makers; i.e., which equilibrium is more informative (has smaller cost).
\begin{thm} \label{thm:expMoreInformative}
	The most informative equilibrium is reached with the maximum possible number of bins:
	\begin{itemize}
		\item [(i)] for $b<0$, if there are two different equilibria with $K$ and $N$ bins where $N>K$, the equilibrium with $N$ bins is more informative.
		\item [(ii)] for $b>0$, the equilibria with infinitely many bins are the most informative ones.
	\end{itemize} 
\end{thm}
\begin{proof}
	\begin{enumerate}[(i)]
	\item Suppose that there exists an equilibrium with $N$ bins, and the corresponding bin-lengths are $l_1<l_2<\ldots<l_N=\infty$ with bin-edges $0=\mm_0<\mm_1<\ldots<\mm_{N-1}<\mm_N=\infty$. Then, the decoder cost is
	\begin{align}
		J^{d,N} &= \mathbb{E}[(M-U)^2] = \mathbb{E}[(M-\mathbb{E}[M|X])^2] \nn\\
		&= \sum_{i=1}^{N} \mathbb{E}\left[\left(M-\mathbb{E}[M|\mm_{i-1}<M<\mm_i]\right)^2|\mm_{i-1}<M<\mm_i\right] \mathrm{Pr}(\mm_{i-1}<M<\mm_i) \nn\\
		&= \sum_{i=1}^{N} \mathrm{Var}\left(M|\mm_{i-1}<M<\mm_i\right) \mathrm{Pr}(\mm_{i-1}<M<\mm_i) \nn\\
		&= \sum_{i=1}^{N} \left({1\over\lambda^2} - {l_i^2\over\me^{\lambda l_i}+\me^{-\lambda l_i}-2}\right) \left(\me^{-\lambda \mm_{i-1}}\left(1-\me^{-\lambda l_i}\right)\right) \,.
	\end{align}
	Now, consider an equilibrium with $N+1$ bins with bin-lengths $\widetilde{l}_1<\widetilde{l}_2<\ldots<\widetilde{l}_{N+1}=\infty$ and bin-edges $0=\widetilde{m}_0<\widetilde{m}_1<\ldots<\widetilde{m}_{N}<\widetilde{m}_{N+1}=\infty$. The relation between bin-lengths and bin-edges can be expressed as $l_k=\widetilde{l}_{k+1}$ and $\mm_k=\widetilde{m}_{k+1}-\widetilde{l}_1$, respectively, for $k=1,2,\ldots,N$ by Theorem~\ref{thm:expPosBias}. Then, the decoder cost at the equilibrium with $N+1$ bins can be written as
	\begin{align}
		J^{d,N+1} &=  \sum_{i=1}^{N+1} \left({1\over\lambda^2} - {\widetilde{l}_i^2\over\me^{\lambda \widetilde{l}_i}+\me^{-\lambda \widetilde{l}_i}-2}\right) \left(\me^{-\lambda \widetilde{m}_{i-1}}\left(1-\me^{-\lambda \widetilde{l}_i}\right)\right) \nn\\
		&= \left({1\over\lambda^2} - {\widetilde{l}_1^2\over\me^{\lambda \widetilde{l}_1}+\me^{-\lambda \widetilde{l}_1}-2}\right) \left(\me^{-\lambda \widetilde{m}_0}\left(1-\me^{-\lambda \widetilde{l}_1}\right)\right) + \sum_{i=2}^{N+1} \left({1\over\lambda^2} - {\widetilde{l}_i^2\over\me^{\lambda \widetilde{l}_i}+\me^{-\lambda \widetilde{l}_i}-2}\right) \left(\me^{-\lambda \widetilde{m}_{i-1}}\left(1-\me^{-\lambda \widetilde{l}_i}\right)\right) \nn\\
		&= \left({1\over\lambda^2} - {\widetilde{l}_1^2\over\me^{\lambda \widetilde{l}_1}+\me^{-\lambda \widetilde{l}_1}-2}\right) \left(1-\me^{-\lambda \widetilde{l}_1}\right) + \sum_{i=2}^{N+1} \left({1\over\lambda^2} - {l_{i-1}^2\over\me^{\lambda l_{i-1}}+\me^{-\lambda l_{i-1}}-2}\right) \left(\me^{-\lambda (\mm_{i-2}+\widetilde{l}_1)}\left(1-\me^{-\lambda l_{i-1}}\right)\right) \nn\\
		&= \left({1\over\lambda^2} - {\widetilde{l}_1^2\over\me^{\lambda \widetilde{l}_1}+\me^{-\lambda \widetilde{l}_1}-2}\right) \left(1-\me^{-\lambda \widetilde{l}_1}\right) + \me^{-\lambda \widetilde{l}_1}\underbrace{\left(\sum_{i=1}^{N} \left({1\over\lambda^2} - {l_{i}^2\over\me^{\lambda l_{i}}+\me^{-\lambda l_{i}}-2}\right) \left(\me^{-\lambda \mm_{i-1}}\left(1-\me^{-\lambda l_{i}}\right)\right)\right)}_{J^{d,N}} \nn\\
		&\overset{(a)}{<} J^{d,N} \left(1-\me^{-\lambda \widetilde{l}_1}\right) + J^{d,N} \me^{-\lambda \widetilde{l}_1} = J^{d,N} \,.
		\label{eq:expInformativeIneq}
	\end{align}
	Thus, $J^{d,N+1}<J^{d,N}$ is obtained, which implies that the equilibrium with more bins is more informative. Here, (a) follows from the fact below:
	\begin{align}
		J^{d,N} &= \sum_{i=1}^{N} \left({1\over\lambda^2} - {l_i^2\over\me^{\lambda l_i}+\me^{-\lambda l_i}-2}\right) \left(\me^{-\lambda \mm_{i-1}}\left(1-\me^{-\lambda l_i}\right)\right) \nn\\
		&> \sum_{i=1}^{N} \left({1\over\lambda^2} - {\widetilde{l}_1^2\over\me^{\lambda \widetilde{l}_1}+\me^{-\lambda \widetilde{l}_1}-2}\right) \mathrm{Pr}(\mm_{i-1}<m<\mm_i) \nn\\ 
		&= \left({1\over\lambda^2} - {\widetilde{l}_1^2\over\me^{\lambda \widetilde{l}_1}+\me^{-\lambda \widetilde{l}_1}-2}\right) \sum_{i=1}^{N} \mathrm{Pr}(\mm_{i-1}<m<\mm_i) \nn\\
		&= \left({1\over\lambda^2} - {\widetilde{l}_1^2\over\me^{\lambda \widetilde{l}_1}+\me^{-\lambda \widetilde{l}_1}-2}\right) \,,
	\end{align}
	where the inequality holds since $\widetilde{l}_1<l_1<l_2<\ldots<l_N$ and $\varphi(s)\triangleq {s^2\over\me^{\lambda s}+\me^{-\lambda s}-2}$ is a decreasing function of $s$, as shown below:
	\begin{align}
		\varphi(s)&= {s^2\over\me^{\lambda s}+\me^{-\lambda s}-2} = {s^2\me^{\lambda s}\over(\me^{\lambda s}-1)^2} \,,\nn\\
		\varphi^\prime(s)&= {s\me^{\lambda s}(\me^{\lambda s}-1)(2\me^{\lambda s}-\lambda s\me^{\lambda s}-\lambda s-2)\over(\me^{\lambda s}-1)^4} = {s\me^{\lambda s}(\me^{\lambda s}-1)\varrho(\lambda s)\over(\me^{\lambda s}-1)^4} \overset{(a)}{<} 0 \,,
	\end{align}
	where (a) follows from \eqref{eq:rhoNegative}. 
	
	\item Now consider an equilibrium with infinitely many bins. By Theorem~\ref{thm:expInfinite}, bin-lengths are $l_1=l_2=\ldots=l^*$, where $l^*$ is the fixed-point solution of the recursion in \eqref{eq:compactRecursion}; i.e., $g(l^*)={2\over\lambda}+2b-h(l^*)$, and bin-edges are $\mm_k = k l^*$. Then, the decoder cost is
	\begin{align}
		J^{d,\infty} &=  \sum_{i=1}^{\infty} \left({1\over\lambda^2} - {(l^*)^2\over\me^{\lambda l^*}+\me^{-\lambda l^*}-2}\right) \left(\me^{-\lambda (i-1) l^*}\left(1-\me^{-\lambda l^*}\right)\right) \nn\\
		&= \left({1\over\lambda^2} - {(l^*)^2\over\me^{\lambda l^*}+\me^{-\lambda l^*}-2}\right) \left(1-\me^{-\lambda l^*}\right) +\sum_{i=2}^{\infty} \left({1\over\lambda^2} - {(l^*)^2\over\me^{\lambda l^*}+\me^{-\lambda l^*}-2}\right) \left(\me^{-\lambda (i-1) l^*}\left(1-\me^{-\lambda l^*}\right)\right) \nn\\
		&= \left({1\over\lambda^2} - {(l^*)^2\over\me^{\lambda l^*}+\me^{-\lambda l^*}-2}\right) \left(1-\me^{-\lambda l^*}\right) +\me^{-\lambda l^*}\sum_{i=2}^{\infty} \left({1\over\lambda^2} - {(l^*)^2\over\me^{\lambda l^*}+\me^{-\lambda l^*}-2}\right) \left(\me^{-\lambda (i-2) l^*}\left(1-\me^{-\lambda l^*}\right)\right) \nn\\
		&= \left({1\over\lambda^2} - {(l^*)^2\over\me^{\lambda l^*}+\me^{-\lambda l^*}-2}\right) \left(1-\me^{-\lambda l^*}\right) +\me^{-\lambda l^*}\underbrace{\sum_{i=1}^{\infty} \left({1\over\lambda^2} - {(l^*)^2\over\me^{\lambda l^*}+\me^{-\lambda l^*}-2}\right) \left(\me^{-\lambda (i-1) l^*}\left(1-\me^{-\lambda l^*}\right)\right)}_{J^{d,\infty}} \nn\\
		\Rightarrow & J^{d,\infty}\left(1-\me^{-\lambda l^*}\right) = \left({1\over\lambda^2} - {(l^*)^2\over\me^{\lambda l^*}+\me^{-\lambda l^*}-2}\right) \left(1-\me^{-\lambda l^*}\right) \nn\\
		\Rightarrow & J^{d,\infty} =  \left({1\over\lambda^2} - {(l^*)^2\over\me^{\lambda l^*}+\me^{-\lambda l^*}-2}\right) \,.
	\end{align}
	Since bin-lengths at the equilibria with finitely many bins are greater than $l^*$ by Theorem~\ref{thm:expInfinite}, and due to a similar reasoning in \eqref{eq:expInformativeIneq} (indeed, by replacing $\widetilde{l}_1$ with $l^*$), $J^{d,\infty}<J^{d,N}$ can be obtained for any finite $N$. Actually, $J^{d,N}$ is a monotonically decreasing sequence with limit $\lim_{N\to\infty} J^{d,N} = J^{d,\infty}$. Thus, the lowest equilibrium cost is achieved with infinitely many bins.
	\end{enumerate}
\end{proof}

Theorem~\ref{thm:expMoreInformative} implies that, among multiple equilibria, the one with the maximum number of bins must be chosen by the players; i.e., the payoff dominant equilibrium is selected \cite{eqSelectionBook}.

\section{Gaussian Distribution}

Let $M$ be a Gaussian r.v. with mean $\mu$ and variance $\sigma^2$; i.e., $M\sim\mathcal{N}(\mu,\sigma^2)$. Let $\phi(m)={1\over\sqrt{2\pi}}\me^{-{m^2\over2}}$ be the PDF of a standard Gaussian r.v., and let $\Phi(b)=\int_{-\infty}^{b}\phi(m){\rm{d}}m$ be its cumulative distribution function (CDF). Then, the expectation of a truncated Gaussian r.v. is the following:
\begin{fact}
	The mean of a Gaussian r.v. $M\sim\mathcal{N}(\mu,\sigma^2)$ truncated to the interval $[a,b]$ is $\mathbb{E}[M|a<M<b]=\mu - \sigma{\phi({b-\mu\over\sigma})-\phi({a-\mu\over\sigma})\over\Phi({b-\mu\over\sigma})-\Phi({a-\mu\over\sigma})}$.
\end{fact}
\begin{proof}
	\begin{align}
		\mathbb{E}[M|a<M<b] &= \int_{a}^{b} m {{1\over\sqrt{2\pi}\sigma}\me^{-{(m-\mu)^2\over2\sigma^2}}\over\int_{a}^{b}{1\over\sqrt{2\pi}\sigma}\me^{-{(m-\mu)^2\over2\sigma^2}}}\mathrm{d}m = {\int_{a}^{b} m {1\over\sqrt{2\pi}\sigma}\me^{-{(m-\mu)^2\over2\sigma^2}}\mathrm{d}m\over\int_{a}^{b}{1\over\sqrt{2\pi}\sigma}\me^{-{(m-\mu)^2\over2\sigma^2}}\mathrm{d}m} \nn\\
		& \stackrel{\mathclap{\substack{{s=(m-\mu)/\sigma} \\ {\mathrm{d}s=\mathrm{d}m/\sigma}}}}{=} \hspace*{1.5em} {\int_{(a-\mu)/\sigma}^{(b-\mu)/\sigma} (\sigma s+\mu) {1\over\sqrt{2\pi}}\me^{-{s^2\over2}}\mathrm{d}s\over\int_{(a-\mu)/\sigma}^{(b-\mu)/\sigma}{1\over\sqrt{2\pi}}\me^{-{s^2\over2}}\mathrm{d}s} = \mu + \sigma{\int_{(a-\mu)/\sigma}^{(b-\mu)/\sigma} s {1\over\sqrt{2\pi}}\me^{-{s^2\over2}}\mathrm{d}s\over\int_{(a-\mu)/\sigma}^{(b-\mu)/\sigma}{1\over\sqrt{2\pi}}\me^{-{s^2\over2}}\mathrm{d}s} \nn\\
		& \stackrel{\mathclap{\substack{{u=s^2/2} \\ {\mathrm{d}u=s\mathrm{d}s}}}}{=} \hspace*{1.5em} \mu + \sigma{\int_{(a-\mu)^2/(2\sigma^2)}^{(b-\mu)^2/(2\sigma^2)}  {1\over\sqrt{2\pi}}\me^{-u}\mathrm{d}u\over\Phi({b-\mu\over\sigma})-\Phi({a-\mu\over\sigma})} = \mu + \sigma{-{1\over\sqrt{2\pi}}\me^{-u}\rvert_{u=(a-\mu)^2/(2\sigma^2)}^{u=(b-\mu)^2/(2\sigma^2)}\over\Phi({b-\mu\over\sigma})-\Phi({a-\mu\over\sigma})} \nn\\
		&= \mu + \sigma{-\phi({b-\mu\over\sigma})+\phi({a-\mu\over\sigma})\over\Phi({b-\mu\over\sigma})-\Phi({a-\mu\over\sigma})} \overset{\substack{{\alpha\triangleq(a-\mu)/\sigma} \\ {\beta\triangleq(b-\mu)/\sigma}}}{=} \mu - \sigma{\phi(\beta)-\phi(\alpha)\over\Phi(\beta)-\Phi(\alpha)} \,.
		\label{eq:gaussCentroid}
\end{align}
\end{proof}
Now we consider an equilibrium with two bins:
\begin{thm}\label{thm:gauss2bins}
	When the source has a Gaussian distribution as $M\sim\mathcal{N}(\mu,\sigma^2)$, there always exists an equilibrium with two bins regardless of the value of $b$.
\end{thm}
\begin{proof}
	Consider the two bins $(-\infty=\mm_0,\mm_1)$ and $[\mm_1,\mm_2=\infty)$. The centroids of the bins (the action of the decoder) are $u_1=\mathbb{E}[M|-\infty<M<\mm_1]=\mu - \sigma{\phi({\mm_1-\mu\over\sigma})\over\Phi({\mm_1-\mu\over\sigma})}$ and $u_2=\mathbb{E}[M|\mm_1\leq M<\infty] = \mu + \sigma{\phi({\mm_1-\mu\over\sigma})\over1-\Phi({\mm_1-\mu\over\sigma})}$. Then, by utilizing \eqref{eq:centroidBoundaryEq}, an equilibrium with two bins exists if and only if
	\begin{align}
		\mm_1 &= {u_1+u_2\over2}+b = {\mu - \sigma{\phi({\mm_1-\mu\over\sigma})\over\Phi({\mm_1-\mu\over\sigma})}+\mu + \sigma{\phi({\mm_1-\mu\over\sigma})\over1-\Phi({\mm_1-\mu\over\sigma})}\over2}+b \nn\\
		&= \mu+{\sigma\over2}\left({\phi({\mm_1-\mu\over\sigma})\over1-\Phi({\mm_1-\mu\over\sigma})}-{\phi({\mm_1-\mu\over\sigma})\over\Phi({\mm_1-\mu\over\sigma})}\right)+b \nn\\
		& \stackrel{\mathclap{c\triangleq{\mm_1-\mu\over\sigma}}}{\Rightarrow} \hspace*{1.5em} \sigma c+\mu = \mu+{\sigma\over2} \left({\phi(c)\over 1-\Phi(c)}-{\phi(c)\over\Phi(c)}\right) + b \nn\\
		& \Rightarrow 2c - {\phi(c)\over 1-\Phi(c)} + {\phi(c)\over\Phi(c)} = {2b\over\sigma} \,.
	\label{eq:gaussTwoBinsEq}
	\end{align}
	Let $f(c)\triangleq 2c - {\phi(c)\over 1-\Phi(c)} + {\phi(c)\over\Phi(c)}$ and $\stackrel{H}{=}$ denote l'H\^ospital's rule, then, observe the following:
	\begin{align}
		\lim\limits_{c\rightarrow-\infty}f(c) &= \lim\limits_{c\rightarrow-\infty} \left(2c - {\phi(c)\over 1-\Phi(c)} + {\phi(c)\over\Phi(c)}\right)\nn\\
		&= \lim\limits_{c\rightarrow-\infty} \left(2c + {\phi(c)\over\Phi(c)}\right) - \cancelto{0}{\lim\limits_{c\rightarrow-\infty}\left({\phi(c)\over 1-\Phi(c)}\right)} \nn\\
		&= \lim\limits_{c\rightarrow-\infty} \left({2c\Phi(c)+\phi(c)\over\Phi(c)}\right) \nn\\
		&\stackrel{H}{=}  \lim\limits_{c\rightarrow-\infty} \left({2\Phi(c)+c\phi(c)\over\phi(c)}\right) \nn\\
		&\stackrel{H}{=}  \lim\limits_{c\rightarrow-\infty} \left({3\phi(c)-c^2\phi(c)\over-c\phi(c)}\right) = \lim\limits_{c\rightarrow-\infty}\left(3-c^2\over-c\right) \stackrel{H}{=} \lim\limits_{c\rightarrow-\infty} 2c\rightarrow-\infty \,,\nn\\
		\lim\limits_{c\rightarrow\infty}f(c) &= \lim\limits_{c\rightarrow\infty} \left(2c - {\phi(c)\over 1-\Phi(c)} + {\phi(c)\over\Phi(c)}\right)\nn\\
		&= \lim\limits_{c\rightarrow\infty} \left(2c - {\phi(c)\over 1-\Phi(c)}\right) + \cancelto{0}{\lim\limits_{c\rightarrow\infty}\left({\phi(c)\over\Phi(c)}\right)} \nn\\
		&= \lim\limits_{c\rightarrow\infty} \left({2c-2c\Phi(c)-\phi(c)\over1-\Phi(c)}\right) \nn\\
		&\stackrel{H}{=}  \lim\limits_{c\rightarrow\infty} \left({2-2\Phi(c)-c\phi(c)\over-\phi(c)}\right) \nn\\
		&\stackrel{H}{=}  \lim\limits_{c\rightarrow\infty} \left({-3\phi(c)+c^2\phi(c)\over c\phi(c)}\right) = \lim\limits_{c\rightarrow\infty}\left(-3+c^2\over c\right) \stackrel{H}{=} \lim\limits_{c\rightarrow\infty} 2c\rightarrow\infty \,,\nn\\
		f^\prime(c) &= 2-{\phi(c)(-c)(1-\Phi(c))-\phi(c)(-\phi(c))\over\left(1-\Phi(c)\right)^2} + {\phi(c)(-c)\Phi(c)-\phi(c)\phi(c)\over\Phi(c)^2} \nn\\
		&= 2-\phi(c)^2\left({1\over\left(1-\Phi(c)\right)^2}+{1\over\Phi(c)^2}\right)+c\phi(c)\left({1\over 1-\Phi(c)}-{1\over\Phi(c)}\right) \nn\\
		&= 2-{\phi(c)\over1-\Phi(c)}\left({\phi(c)\over 1-\Phi(c)}-c\right)-{\phi(c)^2\over\Phi(c)^2}-{c\phi(c)\over\Phi(c)} \,.
	\label{eq:gaussian2BinDerivative}
	\end{align}
	It can be seen that, by using the identities $\phi(c)=\phi(-c)$ and $\Phi(c)=1-\Phi(-c)$, $f^\prime(c)$ is an even function of $c$; i.e., $f^\prime(c)=f^\prime(-c)$. Thus, it can be assumed that $c\geq0$ for the analysis of $f^\prime(c)$. Then, observe the following inequalities:
	\begin{itemize}
		\item In \cite{birnbaumMillsRatio1942}, the inequality on the upper bound of the Mill's ratio is proved as ${\phi(c)\over1-\Phi(c)}<{\sqrt{c^2+4}+c\over2}$. Then,
		\begin{align*}
			{\phi(c)\over1-\Phi(c)}\left({\phi(c)\over 1-\Phi(c)}-c\right) &< 	{\sqrt{c^2+4}+c\over2}\left({\sqrt{c^2+4}+c\over2}-c\right) \nn\\
			&= {\sqrt{c^2+4}+c\over2}{\sqrt{c^2+4}-c\over2} =1 \,.
		\end{align*}	
		\item Since $\mathbb{E}[X|-\infty<X<c]=-{\phi(c)\over\Phi(c)}$ for standard normal distribution, ${\phi(c)\over\Phi(c)}$ is a decreasing function of $c$, and for $c>0$, ${\phi(c)\over\Phi(c)}<{\phi(0)\over\Phi(0)}=\sqrt{2\over\pi}$.
		\item Let $g(c)\triangleq{c\phi(c)\over\Phi(c)}$, then $g^\prime(c)={\phi(c)\left((1-c^2)\Phi(c)-c\phi(c)\over(\Phi(c))^2\right)}$. If we let $h(c)\triangleq(1-c^2)\Phi(c)-c\phi(c)$, then $h^\prime(c)=-2c\Phi(c)+(1-c^2)\phi(c)-\phi(c)-c\phi(c)(-c)=-2c\Phi(c)<0$ for $c>0$. Thus, $g^{\prime\prime}(c)<0$b holds, which implies that $g(c)$ is a concave function of $c$, and takes its maximum value at $g(c^*)$ which satisfies $g^\prime(c^*)=h(c^*)=0$. By solving numerically, we obtain $c^*\simeq0.9557$ and $g(c^*)\simeq0.2908$.
	\end{itemize}
	By utilizing the results above, \eqref{eq:gaussian2BinDerivative} becomes
	\begin{align}
		f^\prime(c) > 2 - 1 - {2\over\pi} - 0.2908 \simeq 0.0726>0 \,.
	\end{align}
	Thus, $f(c)$ is a monotone increasing function and it takes values between $(-\infty,\infty)$; thus, \eqref{eq:gaussTwoBinsEq} has always a unique solution to $f(c)={2b\over\sigma}$. This assures that, there always exists an equilibrium with two bins regardless of the value of $b$. Further, since $f(0) = 2\times 0 - {\phi(0)\over 1-\Phi(0)} + {\phi(0)\over\Phi(0)}=0$, the signs of $b$ and $c$ must be the same; i.e., if $b<0$, the boundary between two bins is smaller than the mean ($\mm_1<\mu$); if $b>0$, the boundary between two bins is greater than the mean ($\mm_1>\mu$). 
\end{proof}

Since the PDF of a Gaussian r.v. is symmetrical about its mean $\mu$, and monotonically decreasing in the interval $[\mu,\infty)$, the following can be obtained using a similar reasoning as in Proposition~\ref{prop1}:
\begin{prop} \label{prop:gaussMonotonicity}
	Suppose there is an equilibrium with $N$ bins for a Gaussian source $M\sim\mathcal{N}(\mu,\sigma^2)$. Then,
	\begin{enumerate} 
		\item[(i)] if $b<0$, bin-lengths are monotonically increasing and the number of bins are upper bounded in the interval $[\mu,\infty)$,
		\item[(ii)] if $b>0$, bin-lengths are monotonically decreasing and the number of bins are upper bounded in the interval $(-\infty,\mu]$.
	\end{enumerate}
\end{prop}
\begin{proof}
	\begin{itemize}
		\item[(i)] 	Consider an equilibrium with $N$ bins for a Gaussian source $M\sim\mathcal{N}(\mu,\sigma^2)$: the $k$-th bin is $[\mm_{k-1},\mm_k)$, and the centroid of the $k$-th bin (i..e, the corresponding action of the decoder) is $u_k=\mathbb{E}[M|\mm_{k-1}\leq M<\mm_k]$ so that $-\infty=\mm_0<u_1<\mm_1<u_2<\mm_2<\ldots<\mm_{N-2}<u_{N-1}<\mm_{N-1}<u_N<\mm_N=\infty$. Further, assume that $\mu$ is in the $t$-th bin; i.e., $\mm_{t-1}\leq \mu<\mm_t$. Due to the nearest neighbor condition (the best response of the encoder) we have $u_{k+1}-\mm_k = (\mm_k - u_k) - 2b$; and due to the centroid condition (the best response of the decoder), we have $u_k=\mathbb{E}[M|\mm_{k-1}\leq M<\mm_k]=\mu - \sigma{\phi({\mm_k-\mu\over\sigma})-\phi({\mm_{k-1}-\mu\over\sigma})\over\Phi({\mm_k-\mu\over\sigma})-\Phi({\mm_{k-1}-\mu\over\sigma})}$. Then, for any bin in $[\mu,\infty)$, since $\mm_{k}>\mu$, the following holds:
		\begin{align}
		u_{k} - \mm_{k-1} &= \mathbb{E}[M|\mm_{k-1}\leq M<\mm_{k}]- \mm_{k-1} \nn\\
		&<\mathbb{E}[M|\mm_{k-1}\leq M<\infty]- \mm_{k-1}=\mu + \sigma{\phi({\mm_{k-1}-\mu\over\sigma})\over1-\Phi({\mm_{k-1}-\mu\over\sigma})}	- \mm_{k-1} \nn\\
		& \stackrel{\mathclap{\text{(a)}}}{<} \mu + \sigma {\sqrt{\left({\mm_{k-1}-\mu\over\sigma}\right)^2+4}+{\mm_{k-1}-\mu\over\sigma}\over2}- \mm_{k-1} \nn\\
		&= {\sigma\over2}\left(\sqrt{\left({\mm_{k-1}-\mu\over\sigma}\right)^2+4}-{\mm_{k-1}-\mu\over\sigma}\right) \nn\\
		&< {\sigma\over2}\left(\sqrt{\left({\mm_{k-1}-\mu\over\sigma}\right)^2+4\left({\mm_{k-1}-\mu\over\sigma}\right)+4}-{\mm_{k-1}-\mu\over\sigma}\right) \nn\\
		&= {\sigma\over2}\left({\mm_{k-1}-\mu\over\sigma}+2-{\mm_{k-1}-\mu\over\sigma}\right) = \sigma \,.
		\label{eq:gaussRightIneq}
		\end{align}
		Here, (a) in due to an inequality on the upper bound of the Mill's ratio \cite{birnbaumMillsRatio1942}. Now, observe the following:
		\begin{align*}
		\sigma > u_{N} - \mm_{N-1} 	&= (\mm_{N-1} - u_{N-1}) - 2b\\
		&> (u_{N-1} - \mm_{N-2}) - 2b\\
		&=  (\mm_{N-2} - u_{N-2}) - 2(2b)\\
		&\vdots\\
		&> u_{t+1}-\mm_t - (N-t-1)(2b)\\
		&= \mm_t-u_t - (N-t)(2b) \\
		&> - (N-t)(2b)\,,
		\end{align*}
		where the inequalities follow from the fact that the Gaussian PDF of $M$ is monotonically decreasing on $[\mu,\infty)$. Hence, for $b<0$, $N-t<-{\sigma\over2b}$ is obtained, which implies that the number of bins in $[\mu,\infty)$ is bounded by $\big\lfloor -{\sigma\over2b} \big \rfloor$. Further, when $b<0$, the following relation holds for bin-lengths:
		\begin{align}
		l_k &= \mm_k-\mm_{k-1} = (\mm_k-u_k) + (u_k-\mm_{k-1}) > (u_k-\mm_{k-1}) + (u_k-\mm_{k-1}) \nn\\
		&= (\mm_{k-1}-u_{k-1}-2b) + (\mm_{k-1}-u_{k-1}-2b) \nn\\
		&> (\mm_{k-1}-u_{k-1}-2b) + (u_{k-1}-\mm_{k-2}-2b) \nn\\
		&= (\mm_{k-1}-u_{k-1}) + (u_{k-1}-\mm_{k-2}) -4b = \mm_{k-1}-\mm_{k-2} - 4b = l_{k-1}-4b \nn\\
		&\Rightarrow l_k > l_{k-1} \,.
		\end{align}
		Thus, the bin-lengths are monotonically increasing in the interval $[\mu,\infty)$ when $b<0$.
		\item[(ii)] Similarly, for any bin in $(-\infty,\mu]$, since $\mm_{k}<\mu$, the following holds:
		\begin{align}
		\mm_k - u_{k} &= \mm_k - \mathbb{E}[M|\mm_{k-1}<M<\mm_k]\nn\\
		&<\mm_k - \mathbb{E}[M|-\infty<M<\mm_k]= \mm_k - \mu + \sigma{\phi({\mm_{k}-\mu\over\sigma})\over\Phi({\mm_{k}-\mu\over\sigma})} \nn\\
		& \stackrel{\mathclap{\text{(a)}}}{=} \sigma{\phi({\mu-\mm_{k}\over\sigma})\over1-\Phi({\mu-\mm_{k}\over\sigma})}-\sigma{\mu-\mm_{k}\over\sigma} \nn\\
		& \stackrel{\mathclap{\text{(b)}}}{<} \sigma \left({\sqrt{\left({\mu-\mm_{k}\over\sigma}\right)^2+4}+{\mu-\mm_{k}\over\sigma}\over2}-{\mu-\mm_{k}\over\sigma}\right)  \nn\\
		&={\sigma\over2}\left({\sqrt{\left({\mu-\mm_{k}\over\sigma}\right)^2+4}-{\mu-\mm_{k}\over\sigma}}\right) \nn\\
		&< {\sigma\over2}\left(\sqrt{\left({\mu-\mm_{k}\over\sigma}\right)^2+4\left({\mu-\mm_{k}\over\sigma}\right)+4}-{\mu-\mm_{k}\over\sigma}\right) \nn\\
		&= {\sigma\over2}\left({\mu-\mm_{k}\over\sigma}+2-{\mu-\mm_{k}\over\sigma}\right) = \sigma \,.
		\label{eq:gaussLeftIneq}
		\end{align}
		Here, (a) holds since $\phi(x)=\phi(-x)$ and $\Phi(x)=1-\Phi(-x)$, and (b) follows from an inequality on the upper bound of the Mill's ratio \cite{birnbaumMillsRatio1942}. Now, observe the following:
		\begin{align*}
		\sigma > \mm_1 - u_1 &= u_2 - \mm_1 + 2b\\
		&> \mm_2 - u_2 + 2b\\
		&=  u_3 - \mm_2 + 2(2b)\\
		&\vdots\\
		&> \mm_{t-1} - u_{t-1} + (t-2)(2b)\\
		&= u_t - \mm_{t-1} + (t-1)(2b) \\
		&> (t-1)(2b)\,,
		\end{align*}
		where the inequalities follow from the fact that the Gaussian PDF of $M$ is monotonically increasing on $(-\infty,\mu]$. Hence, for $b>0$, $t-1<{\sigma\over2b}$ is obtained, which implies that the number of bins in $(-\infty,\mu]$ is bounded by $\big\lfloor {\sigma\over2b} \big \rfloor$. Further, when $b>0$, the following relation holds for bin-lengths:
		\begin{align}
		l_k &= \mm_k-\mm_{k-1} = (\mm_k-u_k) + (u_k-\mm_{k-1}) > (\mm_k-u_k) + (\mm_k-u_k) \nn\\
		&= (u_{k+1}-\mm_k+2b) + (u_{k+1}-\mm_k+2b) \nn\\
		&> (\mm_{k+1}-u_{k+1}+2b) + (u_{k+1}-\mm_k+2b) \nn\\
		&= (\mm_{k+1}-u_{k+1}) + (u_{k+1}-\mm_{k}) + 4b = \mm_{k+1}-\mm_{k} + 4b = l_{k+1}+ 4b \nn\\
		&\Rightarrow l_k < l_{k+1} \,.
		\end{align}
		Thus, the bin-lengths are monotonically decreasing in the interval $(-\infty,\mu]$ when $b>0$.
	\end{itemize}
\end{proof}

After showing that there always exists an equilibrium with two bins independent of $b$, we may ask whether there always exists an equilibrium with $N$ bins, or infinitely many bins. The following theorem answers the second part of this question:
\begin{thm}\label{thm:gaussInfinity}
	For the Gaussian source $M\sim\mathcal{N}(\mu,\sigma^2)$, there exist equilibria with infinitely many bins. 
\end{thm}
\begin{proof}
The proof requires individual analysis for the positive and the negative $b$ values. Firstly, assume a positive bias term; i.e., $b>0$. Now consider a bin on $[\mu,\infty)$; i.e., the bin is the interval $[\mm_{k-1},\mm_k)$ with $\mu\leq\mm_{k-1}<\mm_{k}$. Let $u_k\triangleq \mathbb{E}[M|\mm_{k-1}<M<\mm_{k}]$, then $u_k-\mm_{k-1} < \mathbb{E}[M|\mm_{k-1}<M<\infty] -\mm_{k-1} < \sigma$, where the second inequality follow from \eqref{eq:gaussRightIneq} in Proposition~\ref{prop:gaussMonotonicity}. Further, \eqref{eq:centroidBoundaryEq} implies $\mm_k-u_k=u_{k+1}-\mm_k+2b \Rightarrow 2b+\sigma>\mm_k-u_k > 2b$. Thus, the length of the bin, $l_k\triangleq\mm_k-\mm_{k-1}=(\mm_{k}-u_k)+(u_k-\mm_{k-1})$, is between $2b<l_k<2b+2\sigma$. Similarly, if the bin is on $(-\infty,\mu]$; i.e., $[\mm_{s-1},\mm_s)$ with $\mm_{s-1}<\mm_{s}<\mu$ and $u_s\triangleq \mathbb{E}[m|\mm_{s-1}<m<\mm_{s}]$, it holds that $m_s-u_s<m_s-\mathbb{E}[m|-\infty<m<\mm_{s}]<\sigma$, where the second inequality follow from Proposition~\ref{prop:gaussMonotonicity}. Further, \eqref{eq:centroidBoundaryEq} implies $\mm_s-u_s=u_{s+1}-\mm_s+2b \Rightarrow \mm_s-u_s > 2b$ and $u_{s+1}-\mm_s<\sigma-2b$ (indeed, if $\sigma<2b$, there can be at most one bin on $(-\infty,\mu]$, which results in an absence of a bin-edge on $(-\infty,\mu]$). Thus, the length of the bin, $l_s\triangleq\mm_s-\mm_{s-1}=(\mm_{s}-u_s)+(u_s-\mm_{s-1})$, is between $2b<l_s<2\sigma-2b$. If the bin contains mean; i.e., $\mm_{t}\leq\mu<\mm_{t+1}$, we have $\mm_{t+1}-u_{t+1}=u_{t+2}-\mm_{t+1}+2b<\sigma+2b$ and $u_{t+1}-\mm_{t}=\mm_{t}-u_t-2b<\sigma-2b$. Thus, for the corresponding bin-length, $2b<\mm_{t+1}-\mm_{t}<2\sigma$ is obtained. For the left-most bin-edge, $\mm_1$, there are two possibilities:
\begin{enumerate}[(i)]
	\item \underline{$\mm_1<\mu$} : There can be at most $\big\lfloor {\sigma\over2b} \big \rfloor$ bins with maximum length $2\sigma-2b$, and there is a bin which contains $\mu$, thus $\mm_1>\mu-\big\lfloor {\sigma\over2b} \big \rfloor (2\sigma-2b)-2\sigma$.
	\item \underline{$\mm_1\geq\mu$} : Since $u_2-\mm_1<\sigma$ by \eqref{eq:gaussRightIneq}, $\mm_1-u_1=u_2-\mm_1+2b<2b+\sigma$ is obtained. Further we have $u_1=\mathbb{E}[M|-\infty<M<\mm_1]<\mu$, which implies that $\mm_1$ has an upper bound as $\mm_1<\mu+2b+\sigma$.
\end{enumerate}
Thus, at the equilibrium, the value of the left-most bin-edge is lower and upper bounded as $\mu-\big\lfloor {\sigma\over2b} \big \rfloor (2\sigma-2b)-2\sigma\leq\mm_1\leq\mu+2b+\sigma$ (Note that, here, non-strict inequalities are preferred over the strict ones in order to obtain a bounded, convex and compact set, which will be utilized in the fixed-point theorem to show the existence of an equilibrium). Further, all bin lengths, $l_2, l_3, \ldots$ are lower and upper bounded as $2b\leq l_i\leq\max\{2b+2\sigma,2\sigma,2\sigma-2b\}=2b+2\sigma$ for $i=2,3,\ldots$. Observe that the set $\mathscr{K} \triangleq \left[\mu-\big\lfloor {\sigma\over2b} \big \rfloor (2\sigma-2b)-2\sigma, \mu+2b+\sigma\right] \times [2b, 2b+2\sigma] \times [2b, 2b+2\sigma] \times \cdots$ (where $\{\mm_1, l_2, l_3, \cdots\} \in \mathscr{K}$) is a convex and compact set by Tychonoff's theorem \cite{fixedPointBook}, and the other bin-edges can be represented by $\mm_i=\mm_1+\sum_{j=2}^{i}l_j$. Hence, another convex and compact set $\widehat{\mathscr{K}}$ can be defined such that $\{\mm_1, \mm_2, \cdots\} \in \widehat{\mathscr{K}}$. Further, at the equilibrium, the best responses of the encoder and the decoder in \eqref{centroid} and \eqref{eq:centroidBoundaryEq} can be combined to define a mapping as follows:
\begin{align}
	\mathbf{\mm}\triangleq\begin{bmatrix}
		\mm_{1} \\
		\mm_{2} \\
		\vdots \\
		\mm_{k}\\
		\vdots
	\end{bmatrix} &= 
	\begin{bmatrix}
		{\mathbb{E}[M|\mm_{0}<M<\mm_{1}]+\mathbb{E}[M|\mm_{1}<M<\mm_{2}]\over2}+b \\
		{\mathbb{E}[m|\mm_{1}<M<\mm_{2}]+\mathbb{E}[M|\mm_{2}<M<\mm_{3}]\over2}+b \\
		\vdots \\
		{\mathbb{E}[m|\mm_{k-1}<M<\mm_{k}]+\mathbb{E}[M|\mm_{k}<M<\mm_{k+1}]\over2}+b \\
		\vdots
	\end{bmatrix} \triangleq\mathscr{T}(\mathbf{\mm})\,,
	\label{eq:gaussFixedPoint}
\end{align} 
Note that the mapping $\mathscr{T}(\mathbf{\mm}):\widehat{\mathscr{K}}\rightarrow\widehat{\mathscr{K}}$ is continuous under the point-wise convergence (since, for $\lim_{n\to\infty}\mm_{i_n}=\mm_{i}$, $\lim_{n\to\infty}{\mathbb{E}[M|\mm_{i-1}<M<\mm_{i_n}]+\mathbb{E}[M|\mm_{i_n}<M<\mm_{i+1}]\over2}+b=\lim_{n\to\infty}{\mu - \sigma{\phi({\mm_{i_n}-\mu\over\sigma})-\phi({\mm_{i-1}-\mu\over\sigma})\over\Phi({\mm_{i_n}-\mu\over\sigma})-\Phi({\mm_{i-1}-\mu\over\sigma})}+\mu - \sigma{\phi({\mm_{i+1}-\mu\over\sigma})-\phi({\mm_{i_n}-\mu\over\sigma})\over\Phi({\mm_{i+1}-\mu\over\sigma})-\Phi({\mm_{i_n}-\mu\over\sigma})}\over2}+b={\mu - \sigma{\phi({\mm_{i}-\mu\over\sigma})-\phi({\mm_{i-1}-\mu\over\sigma})\over\Phi({\mm_{i}-\mu\over\sigma})-\Phi({\mm_{i-1}-\mu\over\sigma})}+\mu - \sigma{\phi({\mm_{i+1}-\mu\over\sigma})-\phi({\mm_{i}-\mu\over\sigma})\over\Phi({\mm_{i+1}-\mu\over\sigma})-\Phi({\mm_{i}-\mu\over\sigma})}\over2}+b={\mathbb{E}[M|\mm_{i-1}<M<\mm_{i}]+\mathbb{E}[M|\mm_{i}<M<\mm_{i+1}]\over2}+b$, and this analysis can be generalized to the vector case to get the desired result), and hence, under the product topology (the result follows by incrementing the dimension of the product one-by-one and showing the continuity at each step). Further, since (countably) infinite product of real intervals is a locally convex vector space, $\widehat{\mathscr{K}}$ is a bounded, convex and compact and locally convex space. Hence, there exists a fixed point for the mapping $\mathscr{T}$ such that $\mathbf{\mm^*}=\mathscr{T}(\mathbf{\mm^*})$ by Tychonoff's fixed-point theorem \cite{fixedPointBook}. This proves that there exists an equilibrium with infinitely many bins. 

	Now assume a negative bias term; i.e., $b<0$ and consider a bin on $(-\infty,\mu]$; i.e., the bin is the interval $[\mm_{k-1},\mm_k)$ with $\mm_{k-1}<\mm_{k}<\mu$. Let $u_k\triangleq \mathbb{E}[M|\mm_{k-1}<M<\mm_{k}]$, then $\mm_{k}-u_k < \mm_{k}-\mathbb{E}[M|-\infty<M<\mm_{k}] < \sigma$, where the second inequality follow from \eqref{eq:gaussLeftIneq} in Proposition~\ref{prop:gaussMonotonicity}. Further, \eqref{eq:centroidBoundaryEq} implies $u_{k+1}-\mm_k=\mm_k-u_k-2b \Rightarrow \sigma-2b>u_{k+1}-\mm_k > - 2b$. Thus, the length of the bin, $l_k\triangleq\mm_k-\mm_{k-1}=(\mm_{k}-u_k)+(u_k-\mm_{k-1})$, is between $-2b<l_k<2\sigma-2b$. Similarly, if the bin is on $[\mu,\infty)$; i.e., $[\mm_{s-1},\mm_s)$ with $\mu\leq\mm_{s-1}<\mm_{s}$ and $u_s\triangleq \mathbb{E}[M|\mm_{s-1}<M<\mm_{s}]$, it holds that $u_s-\mm_{s-1}<\mathbb{E}[M|\mm_{s-1}<M<\infty]-\mm_{s-1}<\sigma$, where the second inequality follow from Proposition~\ref{prop:gaussMonotonicity}. Further, \eqref{eq:centroidBoundaryEq} implies $u_{s+1}-\mm_s=\mm_s-u_s-2b \Rightarrow u_{s+1}-\mm_s>-2b$ and $\mm_s-u_s < \sigma + 2b$ (indeed, if $\sigma<-2b$, there can be at most one bin on $[\mu,\infty)$, which results in an absence of a bin-edge on $[\mu,\infty)$). Thus, the length of the bin, $l_s\triangleq\mm_s-\mm_{s-1}=(\mm_{s}-u_s)+(u_s-\mm_{s-1})$, is between $-2b<l_s<2\sigma+2b$. If the bin contains mean; i.e., $\mm_{t}\leq\mu<\mm_{t+1}$, we have $\mm_{t+1}-u_{t+1}=u_{t+2}-\mm_{t+1}+2b<\sigma+2b$ and $u_{t+1}-\mm_{t}=\mm_{t}-u_t-2b<\sigma-2b$. Thus, for the corresponding bin-length, $-2b<\mm_{t+1}-\mm_{t}<2\sigma$ is obtained. For the right-most bin-edge, $\mm_r$, there are two possibilities:
	\begin{enumerate}[(i)]
		\item \underline{$\mm_r\geq\mu$} : There can be at most $\big\lfloor -{\sigma\over2b} \big \rfloor$ bins with maximum length $2\sigma+2b$, and there is a bin which contains $\mu$, thus $\mm_r<\mu-\big\lfloor {\sigma\over2b} \big \rfloor (2\sigma+2b)+2\sigma$.
		\item \underline{$\mm_r<\mu$} : Since $\mm_r-u_r<\sigma$ by \eqref{eq:gaussLeftIneq}, $u_{r+1}-\mm_r=\mm_r-u_r-2b<\sigma-2b$ is obtained. Further we have $u_{r+1}=\mathbb{E}[M|\mm_r<M<\infty]>\mu$, which implies that $\mm_r$ has a lower bound as $\mm_r>\mu+2b-\sigma$.
	\end{enumerate}
	Thus, at the equilibrium, the value of the right-most bin-edge is lower and upper bounded as $\mu+2b-\sigma\leq\mm_r\leq\mu-\big\lfloor {\sigma\over2b} \big \rfloor (2\sigma+2b)+2\sigma$ (Note that, here, non-strict inequalities are preferred over the strict ones in order to obtain a bounded, convex and compact set, which will be utilized in the fixed-point theorem to show the existence of an equilibrium). Further, all bin lengths, $l_{r-1}, l_{r-2}, \ldots$ are lower and upper bounded as $-2b\leq l_i\leq\max\{2\sigma-2b, 2\sigma, 2\sigma+2b\}=2\sigma-2b$ for $i=r-1,r-2,\ldots$. Based on the right-most bin-edge and the bin-lengths, the other bin-edges can be represented by $\mm_{r-i}=\mm_r-\sum_{j=1}^{i}l_{r-j}$. Similar to the previous case, the set $\{\mm_r, l_{r-1}, l_{r-2}, \cdots\} \in \left[\mu+2b-\sigma,\mu-\big\lfloor {\sigma\over2b} \big \rfloor (2\sigma+2b)+2\sigma\right] \times [-2b, 2\sigma-2b] \times [-2b, 2\sigma-2b] \times \cdots$ is a convex and compact set by Tychonoff's theorem \cite{fixedPointBook}. After defining a mapping similar to that in \eqref{eq:gaussFixedPoint}, which is continuous under the point-wise convergence, and hence, under the product topology. Then, there exists a fixed point by Tychonoff's fixed-point theorem \cite{fixedPointBook}, which implies the existence of an equilibrium with infinitely many bins. 
\end{proof}

\begin{rem}
At the equilibrium with infinitely many bins, as the bin-edges get very large in absolute value (i.e., $m_i\rightarrow\infty$ for $b>0$ and $m_i\rightarrow-\infty$ for $b<0$), bin-lengths converge to $2|b|$.
\end{rem}
\begin{proof}
For $b>0$, we can characterize what the bins looks like as the bin-edges get very large with the following analysis:
\begin{align}
\lim_{i\to\infty} &\mathbb{E}[M|\mm_{i-1}^*<M<\mm_{i}^*] - \mm_{i-1}^* = \lim_{i\to\infty} \mathbb{E}[M|\mm_{i-1}^*<M<\mm_{i-1}^*+l_i^*] - \mm_{i-1}^*\nn\\
=& \lim_{\mm_{i-1}^*\to\infty} \mu - \sigma{\phi({\mm_{i-1}^*+l_i^*-\mu\over\sigma})-\phi({\mm_{i-1}^*-\mu\over\sigma})\over\Phi({\mm_{i-1}^*+l_i^*-\mu\over\sigma})-\Phi({\mm_{i-1}^*-\mu\over\sigma})} - \mm_{i-1}^* \nn\\
\stackrel{H}{=} &\lim_{\mm_{i-1}^*\to\infty} \mu-\mm_{i-1}^* - \sigma {\phi({\mm_{i-1}^*+l_i^*-\mu\over\sigma}){-\mm_{i-1}^*-l_i^*+\mu\over\sigma}{1\over\sigma}-\phi({\mm_{i-1}^*-\mu\over\sigma}){-\mm_{i-1}^*+\mu\over\sigma}{1\over\sigma}\over\phi({\mm_{i-1}^*+l_i^*-\mu\over\sigma}){1\over\sigma}-\phi({\mm_{i-1}^*-\mu\over\sigma}){1\over\sigma}} \nn\\
=& \lim_{\mm_{i-1}^*\to\infty} \mu-\mm_{i-1}^* - {(-\mm_{i-1}^*-l_i^*+\mu)\phi({\mm_{i-1}^*+l_i^*-\mu\over\sigma})-(-\mm_{i-1}^*+\mu)\phi({\mm_{i-1}^*-\mu\over\sigma})\over\phi({\mm_{i-1}^*+l_i^*-\mu\over\sigma})-\phi({\mm_{i-1}^*-\mu\over\sigma})} \nn\\
=& \lim_{\mm_{i-1}^*\to\infty} \mu-\mm_{i-1}^* - \left(-\mm_{i-1}^*+\mu-{l_i^*\phi({\mm_{i-1}^*+l_i^*-\mu\over\sigma})\over\phi({\mm_{i-1}^*+l_i^*-\mu\over\sigma})-\phi({\mm_{i-1}^*-\mu\over\sigma})}\right) \nn\\
=& \lim_{\mm_{i-1}^*\to\infty} \mu-\mm_{i-1}^* + \mm_{i-1}^* - \mu + {l_i^*\over 1 - {\phi({\mm_{i-1}^*-\mu\over\sigma})\over\phi({\mm_{i-1}^*+l_i^*-\mu\over\sigma})}} \stackrel{(a)}{\rightarrow} 0
\label{eq:gaussRightInfLength}
\end{align}	 
Here, (a) follows from  $\lim_{\mm_{i-1}^*\to\infty}{\phi({\mm_{i-1}^*-\mu\over\sigma})\over\phi({\mm_{i-1}^*+l_i^*-\mu\over\sigma})}=\lim_{\mm_{i-1}^*\to\infty}\me^{{-({\mm_{i-1}^*-\mu\over\sigma})^2+({\mm_{i-1}^*+l_i^*-\mu\over\sigma})^2\over2}}\rightarrow\infty$. Then, \eqref{eq:centroidBoundaryEq} reduces to 
\begin{align}
\lim_{i\to\infty} \mm_i-&\mathbb{E}[M|\mm_{i-1}<M<\mm_{i}]= \lim_{i\to\infty} \mathbb{E}[M|\mm_{i}<M<\mm_{i+1}]-\mm_i+2b \nn\\ \Rightarrow&\lim_{i\to\infty} \mm_i-\mm_{i-1}=\lim_{i\to\infty}\mm_i-\mm_i+2b \nn\\
\Rightarrow& \lim_{i\to\infty} \mm_i-\mm_{i-1}=2b \,.
\end{align}
In other words, the distance between the centroid and the lower edge of the bin converges to zero (i.e., the centroid of the bin converges to the left-edge), and length of the bins converge to $2b$. 

Similarly, for $b<0$, we can characterize what the bins looks like as the bin-edges get very large (in absolute value) with the following analysis:
\begin{align}
\lim_{i\to\infty} & \mm_{r-i}^* - \mathbb{E}[M|\mm_{r-i-1}^*<M<\mm_{r-i}^*]  = \lim_{i\to\infty} \mm_{r-i}^* - \mathbb{E}[M|\mm_{r-i}^*-l_{r-i}^*<M<\mm_{r-i}^*] \nn\\
=& \lim_{\mm_{r-i}^*\to-\infty} \mm_{r-i}^* - \mu + \sigma{\phi({\mm_{r-i}^*-\mu\over\sigma})-\phi({\mm_{r-i}^*-l_{r-i}^*-\mu\over\sigma})\over\Phi({\mm_{r-i}^*-\mu\over\sigma})-\Phi({\mm_{r-i}^*-l_{r-i}^*-\mu\over\sigma})} \nn\\
\stackrel{H}{=} &\lim_{\mm_{r-i}^*\to-\infty} \mm_{r-i}^* - \mu + \sigma {\phi({\mm_{r-i}^*-\mu\over\sigma}){-\mm_{r-i}^*+\mu\over\sigma}{1\over\sigma}-\phi({\mm_{r-i}^*-l_{r-i}^*-\mu\over\sigma}){-\mm_{r-i}^*+l_{r-i}^*+\mu\over\sigma}{1\over\sigma}\over\phi({\mm_{r-i}^*-\mu\over\sigma}){1\over\sigma}-\phi({\mm_{r-i}^*-l_{r-i}^*-\mu\over\sigma}){1\over\sigma}} \nn\\
=& \lim_{\mm_{r-i}^*\to-\infty} \mm_{r-i}^* - \mu + {(-\mm_{r-i}^*+\mu)\phi({\mm_{r-i}^*-\mu\over\sigma})-(-\mm_{r-i}^*+l_{r-i}^*+\mu)\phi({\mm_{r-i}^*-l_{r-i}^*-\mu\over\sigma})\over\phi({\mm_{r-i}^*-\mu\over\sigma})-\phi({\mm_{r-i}^*-l_{r-i}^*-\mu\over\sigma})} \nn\\
=& \lim_{\mm_{r-i}^*\to-\infty} \mm_{r-i}^* - \mu + \left(-\mm_{r-i}^*+\mu-{l_{r-i}^*\phi({\mm_{r-i}^*-l_{r-i}^*-\mu\over\sigma})\over\phi({\mm_{r-i}^*-\mu\over\sigma})-\phi({\mm_{r-i}^*-l_{r-i}^*-\mu\over\sigma})}\right) \nn\\
=& \lim_{\mm_{r-i}^*\to-\infty} \mm_{r-i}^* - \mu - \mm_{r-i}^* + \mu - {l_i^*\over  {\phi({\mm_{r-i}^*-\mu\over\sigma})\over\phi({\mm_{r-i}^*-l_{r-i}^*-\mu\over\sigma})}-1} \stackrel{(a)}{\rightarrow} 0 \,.
\end{align}	 
Here, (a) follows from  $\lim_{\mm_{r-i}^*\to-\infty}{\phi({\mm_{r-i}^*-\mu\over\sigma})\over\phi({\mm_{r-i}^*-l_{r-i}^*-\mu\over\sigma})}=\lim_{\mm_{r-i}^*\to\infty}\me^{{-({\mm_{r-i}^*-\mu\over\sigma})^2+({\mm_{r-i}^*-l_{r-i}^*-\mu\over\sigma})^2\over2}}\rightarrow\infty$. Similar to the $b>0$ case, the distance between the centroid and the upper edge of the bin converges to zero (i.e., the centroid of the bin converges to the right-edge), and length of the bins converge to $-2b$. 
\end{proof}
	
\section{Concluding Remarks}
In this paper, the Nash equilibrium of cheap talk has been characterized for exponential and Gaussian sources. For exponential sources, it has been shown that the number of bins is bounded for bias $b<0$, whereas there exist equilibria with infinitely many bins for $b>0$. Further, it has been proved that, as the number of bins increases, the equilibrium cost of the encoder and decoder reduces. For Gaussian sources, there always exists an equilibrium with infinitely many bins.

Future work includes extending the analysis to arbitrary sources with semi-infinite support and two-sided infinite support, and the investigation of upper bounds on the number of bins and the relation between the number of bins and the equilibrium costs of the players (i.e., equilibrium selection problem). Further, the existence and convergence of equilibria; i.e., under what conditions the best responses of the encoder and the decoder match each other, can be analyzed. It is also interesting to analyze the performance loss due to the misalignment between the objective functions in order to obtain comparisons with optimal quantizers.

\section*{Acknowledgment}

Some of the results, in particular Prop.~\ref{prop1} and Thm.~\ref{thm:expPosBias} build on the project report \cite{FurrerReport} written by  Philippe Furrer, Stephen Kerner and Stanislav Fabricius.

\bibliographystyle{IEEEtran}
\bibliography{../SerkanBibliography}

\end{document}